\documentclass[11pt]{article}
\usepackage[utf8]{inputenc}
\usepackage{verbatim}
\usepackage{appendix}
\usepackage{amsmath}
\usepackage{amsfonts}
\usepackage{amssymb}
\usepackage{amsthm}
\usepackage{graphicx}
\usepackage{natbib}
\usepackage{color}
\usepackage{caption}
\usepackage{subcaption}
\usepackage[affil-it]{authblk}
\usepackage{algpseudocode, algorithm, algorithmicx}
\usepackage[left=0.8in,right=0.8in,top=1in,bottom=1in]{geometry}

\newcommand{\beps}{{\boldsymbol \epsilon}}

\def\bibfont{\normalsize}

\newcommand{\x}{\mathbf{x}}
\newcommand{\y}{\mathbf{y}}
\newcommand{\X}{\mathbf{X}}
\newcommand{\argmin}{\operatornamewithlimits{arg\,min}}

\newcommand{\bSigma}{ \boldsymbol{\Sigma}}
\newcommand{\var}{\mathrm{var}}
\newcommand{\E}{\mathbb{E}}
\newcommand{\A}{\mathcal{A}}
\newcommand{\Z}{\mathbf{Z}}
\renewcommand{\b}{\mathbf{b}}
\newcommand{\bsigma}{\boldsymbol{\sigma}}
\newcommand{\bbeta}{\boldsymbol{\beta}}
\newcommand{\btau}{\boldsymbol{\tau}}
\newcommand{\bP}{\mathbf{P}}
\newcommand{\z}{\mathbf{z}}
\renewcommand{\H}{\mathbf{H}}
\newcommand{\supp}{\mathrm{supp}}

\newcommand{\w}{\mathbf{w}}
\renewcommand{\c}{\mathbf{c}}
\providecommand{\keywords}[1]{\textbf{\textit{Keywords---}} #1}

\makeatletter
\newcommand{\distas}[1]{\mathbin{\overset{#1}{\kern\z@\sim}}}%
\makeatother

\newtheorem{theorem}{Theorem}
\newtheorem{proposition}[theorem]{Proposition}
\newtheorem{lemma}{Lemma}[section]

\author{Arend Voorman\thanks{voorma@uw.edu}}
\author{Ali Shojaie}
\author{Daniela Witten}
\affil{University of Washington}

\title{Inference in High Dimensions with the Penalized Score Test}
\begin{document}

\maketitle
\def\spacingset#1{\renewcommand{\baselinestretch}%
{#1}\normalsize} \spacingset{1}
\abstract{
In recent years, there has been considerable theoretical development regarding variable selection consistency of penalized regression techniques, such as the lasso. However, there has been relatively little work on quantifying the uncertainty in these selection procedures. In this paper, we propose a new method for inference in high dimensions using a score test based on penalized regression. In this test, we perform penalized regression of an outcome on all but a single feature, and test for correlation of the residuals with the held-out feature. This procedure is applied to each feature in turn. Interestingly, when an $\ell_1$ penalty is used, the sparsity pattern of the lasso corresponds exactly to a decision based on the proposed test. Further, when an $\ell_2$ penalty is used, the test statistic corresponds precisely to a score statistic in a mixed effects model, in which the effects of all but one feature are assumed to be random. We formulate the hypothesis being tested as a compromise between the null hypotheses tested in simple linear regression on each feature and in multiple linear regression on all features, and develop reference distributions for some well-known penalties. We also examine the behavior of the test on real and simulated data.
}

\keywords{lasso, large $p$ small $n$, model selection, $p$-value, penalized estimation, significance}

\spacingset{1.5}

\section{Introduction}\label{sec:1:intro}
Suppose we are interested in the association between an outcome variable $\y \in \mathbb{R}^n$ and a set of predictors $\x_j \in \mathbb{R}^n,\, j=1, \dots,d$. In order to assess this we might consider the model
\begin{equation}\label{linear1}
\y = \X\bbeta + \beps,
\end{equation}
where $\X = [\x_1,\dots,\x_d]\in \mathbb{R}^{n \times d}$, $\bbeta \in \mathbb{R}^d$ is vector of coefficients, and $\beps \in \mathbb{R}^n$ is a vector of errors with mean zero and constant variance. If the number of variables $d$ is much smaller than $n$, we could perform a formal statistical test for whether an element of $\bbeta$ is zero using classical methods, such as the score, likelihood ratio, or Wald test. However, in the high-dimensional setting, when the number of variables $d$ is large, these tests have low power, or are undefined.

In the case where $d$ is large, penalized regression techniques, such as the lasso \citep{tibs1996}, which takes the form
\begin{equation}\label{lasso1}
\hat \bbeta_\lambda = \argmin_{\b \in \mathbb{R}^d}\left\{\frac{1}{2n}\|\y-\X\b\|_2^2 + \lambda \|\b\|_1 \right\},
\end{equation}
can be used to provide a sparse estimate of $\bbeta$. Under suitable conditions, explored by \cite{ZY2006}, the sparsity pattern of the lasso coefficient vector $\hat \bbeta_\lambda$ correctly identifies which elements of $\bbeta$ are zero. However, the lasso and related procedures do not provide $p$-values or confidence intervals, and thus devising formal hypothesis tests for the parameter $\bbeta$ remains an open problem.

In recent years there has been some work on formal inference in the high-dimensional setting: the bootstrap has been used to estimate the sampling distribution of lasso coefficients \citep{tibs1996, bach2008bolasso, chatterjee2011bootstrapping},   \citet{meinshausen2010stability} proposed \emph{stability selection},  a sub-sampling technique which can control familywise error rates, and \citet{fan2001variable} provided sandwich formulas for penalized coefficients. However, bootstrapping and sub-sampling are computationally expensive and their finite sample properties may not be desirable, especially for methods that involve thresholding, like the lasso. Available variance formulas also suffer shortcomings: they typically neglect the fact that the tuning parameter tends to be selected based on the data, and are often only available for the non-zero coefficients.

Recently, \citet{lockhart2013significance} proposed the \emph{covariance test}, which produces a sequence of $p$-values as $\lambda$ decreases and features become non-zero in the lasso regression \eqref{lasso1}. Suppose the $k^{th}$ feature to become non-zero does so when the tuning parameter in \eqref{lasso1} is $\lambda_k - \eta$, for some arbitrarily small constant $\eta > 0$. Here, $\lambda_k$ is the $k^{th}$ knot in the lasso solution path. The covariance test statistic produces a $p$-value for each $\lambda_k$, which tests the null hypothesis that $\mathrm{supp}(\hat \bbeta_{\lambda_k}) \supseteq \mathrm{supp}(\bbeta)$. Thus, the covariance test can be used to test whether all relevant variables are non-zero in the lasso coefficient vector. However, this does not give confidence intervals or $p$-values for any individual variable's coefficient. 

{Even more recently, \citet{taylor2014post} and \citet{lee2013exact} extended the covariance testing framework to test hypotheses about individual features, after conditioning on a model selected by the lasso. However, their framework permits inference only about features which have non-zero coefficients in a lasso regression; this set of features will vary across samples, making interpretation difficult. In contrast, our framework can be applied to all features in a data set, and can be used to understand why coefficients in a lasso regression are non-zero in the first place.}

Alternatively, \citet{zhang2011confidence}, \citet{van2013asymptotically} and \citet{javanmard2013confidence} proposed the \emph{low-dimensional projection estimator} (LDPE) for inference in high dimensions based on inverting the stationary conditions for lasso regression. To do this, LDPE uses lasso regression among the covariates to estimate the inverse of $\X^T\X$. Under suitable assumptions, LDPE is asymptotically optimal, in that the variance of the estimator achieves the Gauss-Markov lower bound. However, unlike \citet{lockhart2013significance}, who give $p$-values associated with the knots in the lasso solution path, this method uses the lasso as a starting point for a different estimator. Decisions made using their confidence intervals need not correspond to variable selection using the lasso. 

We applied the covariance test and LDPE to a diabetes data set, previously studied by \citet{efron2004least}, and which we analyze in greater detail in Section~\ref{sec:diabetes}. The data consist of a measure of diabetes disease progression for 442 patients, along with 10 variables. Table~\ref{tab:diab} lists the variables introduced or removed at each knot $\lambda_k$ in the lasso solution path, along with their associated $p$-values from the covariance test and LDPE. For the sake of comparison, we also include the $p$-values produced by multiple linear regression with all variables, the $p$-values produced by simple linear regression of the outcome on each feature separately, and the $p$-values from our proposed lasso-penalized score test, described in Section~\ref{sec:method}. LDPE requires specification of a tuning parameter, which we chose using 10-fold cross-validation.

\begin{table}[htbp] 
\spacingset{1}
\footnotesize
\centering
\begin{tabular}{l|l|lll|ll}
Knot & Predictor & covTest  & LDPE & Pen. score test & Multiple lin. reg. & Simple lin. reg.   \\ \hline \hline
1&  BMI &  $3.7 \times 10^{-9}$ & $2.3 \times 10^{-15}$ & $5.1\times 10^{-22}$ &$4.3\times 10^{-14}$ & $3.4 \times 10^{-42}$\\ 
  2& LTG &  $1.9\times 10^{-22}$ & $2.9 \times 10^{-10}$& $2.6\times 10^{-18}$ & $1.6 \times 10^{-5}$ & $8.8 \times 10^{-39}$ \\ 
   3& MAP &   0.005 & $1.4 \times 10^{-6}$ & $2.6 \times 10^{-8}$&$9 \times 10^{-8}$  & $1.6 \times 10^{-22}$  \\ 
   4& HDL &   0.003 & 0.93 &$3.6 \times 10^{-15}$ & 0.63 & $6.1 \times 10^{-18}$\\ 
   5& Sex & 0.008 &$1.7 \times 10^{-4} $&0.002   & 0.36& 0.0012\\ 
   6&Glu &  0.86 &0.31 & 0.057 &0.079 & $7.6 \times 10^{-17}$ \\ 
   7&TC&  0.04 &0.07 &0.14 & 0.058 &$ 6.9 \times 10^{-6}$\\ 
   8&TCH& 0.54 &0.41 & 0.17 & 0.27& $2.3 \times 10^{-21}$ \\ 
  9&LDL & 0.87 &0.30 &0.21 & 0.16& $2.3 \times 10^{-4}$ \\ 
  10&Age & 0.98 & 0.87 &0.87 & 0.87 & $7.1 \times 10^{-5}$\\ 
 11  & -HDL& -&- &   - &-&\\ 
12&HDL  & 0.80 & 0.93&$3.6 \times 10^{-15}$ & 0.63 & $6.1 \times 10^{-18}$ \\
\end{tabular}
\caption{The diabetes data set. Variables are ordered according to when their coefficients become non-zero in the lasso solution \eqref{lasso1}, as $\lambda \rightarrow 0$. `covTest' refers to the method of \citet{lockhart2013significance} and $p$-values for this method were produced by the covTest R package. `LDPE' refers to the method of \citet{zhang2011confidence} and \citet{van2013asymptotically}, for which code was provided by the authors. Mutiple and simple linear regression refer to those  $p$-values  produced by the Wald test. `Pen. score test' refers to the lasso-penalized score test, described in Section~\ref{sec:method}, with $\lambda=4$. At the $11^{th}$ knot, HDL leaves the lasso solution, and no $p$-value is available for covTest. Of the $p$-values presented in this table, only those from covTest are ordered. For ease of display, here we present $p$-values for all methods in the ordering given by the covTest $p$-values.}\label{tab:diab}
\end{table}

Using the covariance test to select $\lambda$, we might decrease $\lambda$ and stop when some $p$-value greater than $0.05$ is observed. Using this rule, we would stop upon reaching a $p$-value of 0.86 at knot 6, and report a model with 5 covariates. However, knot 7 produces a $p$-value of 0.04, which suggests that all variables were not in the model after knot 5. Should we use the model with 7 or 5 variables? In the model with 7 variables, how should we interpret the presence of glucose, which produced a $p$-value of $0.86$?  Further, using any reasonable stopping rule, we would include HDL in our model. However, HDL yields a $p$-value of 0.63 in multiple linear regression, suggesting there is little evidence for its association after accounting for trends in all other features. How should we interpret this discrepancy? The answers to these questions are not clear.

Using LDPE, the results are broadly similar to those from multiple linear regression. Here, the sample size is sufficient so that the inverse of $\X^T\X$ can be accurately estimated. However, in higher dimensions, this may not be the case. 

In this paper, we propose the \emph{penalized score test}, which can be interpreted as a compromise between multiple linear regression on all features, and simple linear regression on each feature separately. We show that the sparsity pattern of the lasso results from a decision based on this test. Unlike the covariance test statistic, it gives $p$-values for the association of each individual feature with the outcome, and unlike LDPE, the resulting $p$-values are directly related to variable selection using the lasso.

The rest of the paper is organized as follows. In Section~\ref{sec:method} we describe our proposed method, the penalized score test, for general penalties. In Section~\ref{sec:lasso} we consider the special case of the lasso-penalized score test, and explore its asymptotic distribution, its relationship to consistent variable selection (Section~\ref{sec:bias}), and the behavior of the test on simulated and real data (Sections~\ref{sec:sim}, \ref{sec:diabetes} and \ref{sec:thresh}). In Section~\ref{sec:ridge} we consider the special case of the ridge-penalized score test. In Section~\ref{sec:noncon} we propose extensions to other sparsity-inducing penalties. We end with a discussion in Section~\ref{choosinglambda}.

\section{The penalized score test}\label{sec:method}
Throughout the paper, we assume that  $\y^T{\bf 1}_n = 0$, $\x_j^T\x_j = n$ and $\x_j^T{\bf 1}_n=0$ for $j=1,\dots,d$. Vectors are denoted in lowercase bold font, while matrices are in uppercase bold font. In order to simplify the notation, we will consider a single variable of interest $\x = \x_j$, and denote $\Z = [\,\x_k, k \neq j\,] \in \mathbb{R}^{n \times (d-1)}$ as the matrix containing all other features.  Note that any procedure applied to $\x_j$ can be applied to all other variables in turn. With this notation, we re-write the model \eqref{linear1} as
\begin{equation}\label{mymodel}
\y = \alpha \x + \Z\bbeta+ \beps,
\end{equation}
where $\alpha \in \mathbb{R}$ and $\bbeta \in \mathbb{R}^{d-1}$.

Our goal is to test $H_0: \alpha =0$.  One way to do this is using the score test, based on the derivative of the log-likelihood of the model, also known as the score, evaluated under the null hypothesis. Using the model \eqref{mymodel}, and assuming normality of the errors, the (scaled) score statistic is
\begin{equation}\label{classic}
T = \x^T(\y - \y^0)/\sqrt{n},
\end{equation}
where $\y^0 = \Z\bbeta$. We reject $H_0$ when $|T|$ is large, with respect to an appropriate reference distribution. 

Typically, in applying the score test, the parameters are estimated under the constraints imposed by the null hypothesis. In the setting of \eqref{mymodel}, this corresponds to estimating $\bbeta$ using multiple linear regression of $\y$ on $\Z$. However, when $d$ is {an appreciable fraction of $n$, multiple linear regression of $\y$ on $\Z$ yields highly variable coefficient estimates, resulting in low power, and when $d > n$ multiple linear regression of $\y$ on $\Z$} is undefined. As an alternative, we could use a small subset of the other features in estimating $\bbeta$, with e.g. step-wise regression. However, selecting an appropriate model is challenging, and inference after model selection is notoriously difficult \citep{berk2012valid,leeb2005model}.

Instead, we propose to estimate $\bbeta$ in \eqref{mymodel} using penalized regression. The proposed approach is as follows. We first calculate $\hat \b^0_\lambda$, which serves as an estimate of $\bbeta$, using the penalized regression
\begin{equation}\label{eq:restricted}
\hat{\b}^0_\lambda = \argmin_{\b}\left\{ \|\y - \Z\b\|_2^2/(2n) + \lambda J(\b) \right\},
\end{equation}
where $J(\b)$ is a penalty function, such as the lasso ($J(\b)=\|\b\|_1$), ridge ($J(\b)=\|\b\|_2^2/2$), or subset selection ($J(\b) = \|\bbeta\|_0$), and $\lambda$ is a non-negative tuning parameter.  We then set $\hat {\bf y}^0_{\lambda} = \Z \hat{\b}^0_\lambda$, and form the test statistic
\begin{equation}\label{eq:teststat}
T_{\lambda} = \x^T(\y-\hat {\bf y}^0_{\lambda})/\sqrt{n},
\end{equation}
a measure of association between $\x$ and $\y -\hat {\bf y}^0_\lambda$. We declare $T_\lambda$ to be statistically significant, for a null hypothesis to be discussed in Section~\ref{whatnull}, when $|T_\lambda|$ is large, based on an appropriate reference distribution. Since $T_{\lambda}$ looks superficially like the score statistic from linear regression \eqref{classic}, we refer to this procedure as the \emph{penalized score test}. 

Interestingly, if we choose $J(\b) = \|\b\|_1$, and declare $T_\lambda$ to be significant when  $|T_{\lambda}| > \sqrt{n}\lambda$, then $T_\lambda$ is significant precisely when $\x$'s coefficient is non-zero in a lasso-penalized regression of $\y$ on $\x$ and $\Z$ together. That is, the sparsity pattern of the lasso solution for a given $\lambda$  is the result of a decision based on the proposed test.  We make this assertion precise in Proposition~\ref{prop:score}, the proof of which follows immediately from the Karush-Kuhn-Tucker conditions for lasso regression. 

\begin{proposition}
\label{prop:score}
Let $(\hat a_\lambda,\hat \b_\lambda)$ be the lasso solution
\begin{equation}\label{lasso2}
(\hat a_\lambda,\hat \b_\lambda) = \argmin_{(a,\b) \in \mathbb{R}^d} \left \{ \|\y - a\x- \Z\b\|_2^2/(2n) + \lambda \|(a,\b)\|_1\right\},
\end{equation}
and let $T_\lambda$ be as in \eqref{eq:restricted} and \eqref{eq:teststat} with $J(\b) = \|\b\|_1$. Then, $\hat a_{\lambda} \neq 0$ if and only if $|T_{\lambda}| >  \sqrt{n}\lambda$.
\end{proposition}

Further, if we choose $J(\b) = \|\b\|_2^2/2$, corresponding to ridge regression, then $T_{\lambda}$ is precisely the score statistic from a mixed effects model, where the effects of $\Z$ are assumed to be random. We make this assertion precise in Proposition~\ref{prop:mixed}. This connection is further explored in Section~\ref{sec:ridge}.

\begin{proposition}
\label{prop:mixed}
Suppose that \begin{align} \label{eq:mixed}
\y \mid \bbeta &= \alpha \x +\Z\bbeta + \beps \nonumber \\
\bbeta &\sim  N_{d-1}\left(0,\sigma^2_\epsilon(n\lambda)^{-1}{\bf I}_{d-1}\right) \\
\beps &\sim N_n(0, \sigma^2_\epsilon {\bf I}_n), \nonumber
\end{align}
and let $l(\alpha)$ be the log-likelihood of $\y$, with score $\dot l(\alpha) =\tfrac{\partial}{\partial \alpha} l(\alpha)$. Let $T_\lambda$ be as in \eqref{eq:restricted} and \eqref{eq:teststat}  with $J(\b) = \|\b\|^2_2/2$. Then $T_\lambda/\sigma^2_\epsilon = \dot l(0)/\sqrt{n}$.
\end{proposition}
\subsection{What hypothesis is being tested?}\label{whatnull}
When using penalized regression to estimate $\bbeta$, some systematic bias is incurred: our estimate $\hat \b_\lambda^0$ is shrunken towards zero, relative to the unbiased multiple linear regression estimate. In this section, we describe how this bias affects inference, for a given tuning parameter $\lambda$. We will see that  the hypothesis being tested using the penalized score test \eqref{eq:teststat} depends on the tuning parameter $\lambda$.

{For the moment, consider the case where $\sqrt{n}\lambda \rightarrow 0$ as $n \rightarrow \infty$ and the dimension $d$ is fixed. As shown by \citet{KF2000}, for bridge penalties ($J(\b) =\sum_j |b_j|^\gamma$ where $\gamma > 0$) we have that $\hat \b_\lambda^0 \rightarrow_p  \hat \b_0^0 \equiv (\Z^T\Z)^{-1}\Z^T\y $, where we recognize  $\hat \b_0^0$ as the multiple linear regression estimate of $\bbeta$ under $H_0: \alpha = 0$. Thus, in this asymptotic setting,  $T_\lambda$ has the same limiting distribution as the classical score statistic, and one can interpret $T_\lambda$ as a test of $H_0: \alpha = 0$ vs. $H_1: \alpha \neq 0$. However, this asymptotic treatment neglects the fact that in any finite sample, $T_\lambda$ depends on the tuning parameter $\lambda$. For this reason, throughout this paper we predominantly consider an asymptotic scenario with $\lambda$ fixed, a scenario also considered by \citet{yu2002penalized} in the context of penalized spline estimation.
}

In the fixed-$\lambda$ regime, we define the population-level parameters
\begin{equation}\label{eq:params}
( a_\lambda,\b_\lambda) = \argmin_{a \in \mathbb{R},\b \in \mathbb{R}^{d-1}}\left\{ 
\frac{1}{2n}\E \|\y - a\x - \Z\b\|_2^2 + \lambda J(\b)
 \right\}.
\end{equation} 
{The stationary conditions of \eqref{eq:params} imply that $a_\lambda = \E[\x^T(\y- \Z\b_\lambda)/n]$, which is a measure of linear association between $\x$ and $\y-\Z\b_\lambda$. That is, $a_\lambda$ is a measure of correlation between the feature of interest $\x$, and the outcome $\y$ with the penalized effects of $\Z$ removed.} Note that when $a_\lambda = 0$, we have that $\E[\, \x^T(\y-\Z\b_\lambda)/\sqrt{n}\,] = 0$, where we recognize $\x^T(\y-\Z\b_\lambda)/\sqrt{n}$ as the penalized score statistic with $\hat \b_\lambda^0$ replaced by $\b_\lambda$. Provided $\hat \b^0_\lambda$ converges quickly enough to $\b_\lambda$, then $T_\lambda$ is centered around zero, asymptotically, when $a_\lambda = 0$.
Thus, the statistic $T_\lambda$ tests
$$H_{0,\lambda} : a_\lambda = 0 \quad \text{vs.} \quad H_{1,\lambda} : a_\lambda \neq 0.$$
We can write the parameter $a_\lambda$ more simply as 
\begin{equation}\label{eq:parama}
a_\lambda = \alpha + \bsigma_{xz}^T(\bbeta-\b_\lambda),
\end{equation}
where $\bsigma_{xz} =\Z^T\x/n$. Recall that in multiple linear regression of $\y$ on $(\x,\Z)$, the parameter associated with $\x$ is $\alpha$, while in simple linear regression of $\y$ on $\x$ alone, the coefficient associated with $\x$ is $\alpha + \bsigma_{xz}^T\bbeta$. 
Now, when $\lambda=0$, then $a_\lambda = \alpha$. On the other hand, when $\lambda$ is large, $\b_\lambda$ tends to zero and $a_\lambda$ tends to $\alpha + \bsigma^T_{xz}\bbeta$, provided $J(\b) >J(0)$ when $\b \neq {\bf 0}$. For moderate values of $\lambda$, $a_\lambda$ is thus a compromise between the multiple and simple linear regression parameters associated with $\x$. 

{
To make this interpretation of the parameter $a_\lambda$ more concrete, consider the case of $J(\b) = \|\b\|_0$, which corresponds to subset selection. This setting was studied by \citet{berk2012valid}, who discuss how the parameter associated with a feature $\x$ depends on which subset of the features $\bf Z$ are included in a regression model. Our framework, in which the parameter associated with $\x$ depends on the extent to which the effects of $\Z$ are penalized, generalizes this concept. }

In this fixed-$\lambda$ regime, the distribution of $T_\lambda$ depends on the choice of the penalty $J(\b)$ and the value of $\lambda$ in \eqref{eq:restricted}. In Section~\ref{sec:lasso} we give asymptotic theory for the distribution of $T_\lambda$ for the special case of lasso regression, $J(\b) = \|\b\|_1$. The distribution of the ridge-penalized score statistic, $J(\b) = \|\b\|_2^2/2$, comes from mixed-model theory, and is given in Section~\ref{sec:ridge}. We briefly discuss other penalty choices in Section~\ref{sec:noncon}.

\section{The lasso-penalized score test}\label{sec:lasso}
In this section, we examine in greater detail the penalized score test when the lasso is chosen as the penalty. That is, we first obtain the penalized coefficient vector
\begin{equation}\label{eq:lassocoef}
\hat \b^0_\lambda = \argmin_{\b \in \mathbb{R}^{d-1}}\left\{\|\y - \Z\b\|_2^2/(2n) + \lambda \|\b\|_1\right\},
\end{equation}
and then form the test statistic $T_\lambda = (\y-\Z\hat \b_\lambda^0)/\sqrt{n}$. Here we state Proposition~\ref{thm:lassodist}, proven in the Appendix, which gives the asymptotic distribution of $T_\lambda$.

First, we require some notation. Let $\A = \mathrm{supp}(\b_\lambda)$, where $|\A| = q$, and assume, without loss of generality, that $\b_\lambda$ is ordered so that $\b_\lambda = (b_{\lambda,1},\dots, b_{\lambda,q},0,\dots,0)^T$, and partition $\Z$ as $\Z = [\Z_{\A},\Z_{\A^c}]$. Denote ${\bf P}_\A =  \Z_\A( \Z^T_\A \Z_\A)^{-1} \Z^T_\A$ as the projection onto the columns of $\Z_\A$, and let $\bSigma_\A = \Z_\A^T\Z_\A/n$.

Note that the stationary conditions of \eqref{eq:params} with $J(\b) =\|\b\|_1$ require that
\begin{equation}\label{eq:stationarybl}
 \lambda \btau=\E[\Z^T(\y -a_\lambda\x- \Z\b_\lambda)/n],
 \end{equation}
for some $\btau$ satisfying $\btau_\A = \mathrm{sign}(\b_{\lambda\A})$ and $\|\btau_{\A^c}\|_\infty \leq 1.$ 

We will require the following conditions. Note that some of these conditions depend on $\lambda$, and thus may hold for some values of $\lambda$, and not for others.
\begin{itemize}
\item[({\bf A1})] $\y =\alpha\x +  \Z\bbeta + \beps$, where $\x$ and $\Z$ are fixed, and $\beps = [\epsilon_1,\dots,\epsilon_n]^T$ are independent and identically distributed with mean zero and variance $\sigma^2_\epsilon$. 
\item[({\bf A2})] The covariates $[\x,\Z_\A]$ are such that $\lim_{n \rightarrow \infty}\| {\bf r} \|_\infty/ \|{\bf r}\|_2 \rightarrow 0$, where ${\bf r} =  ({\bf I}_n - {\bf P}_\A)\x \in \mathbb{R}^n.$
\end{itemize}
Condition ({\bf A2}) is needed in order to apply the Lindeberg-Feller Central Limit Theorem, and requires that no single element of $({\bf I}_n - {\bf P}_\A)\x$ is too large, relative to the other elements.

In order to allow $d$ to grow more quickly than $n$, we require the following additional conditions.
\begin{itemize}
\item[({\bf A3})] $\lambda$, $(\alpha,\bbeta)$ and $(\x,\Z)$ are such that $\|\btau_{\A^c}\|_\infty \leq 1 - \delta$ for some $\delta > 0$.
\item[({\bf A4})] The matrix $\Z$ is such that $\|\Z^T_{\A^c}\Z_\A/n\|_\infty = o(\sqrt{n/ (q \log q)})$.
\item[({\bf A5})] The errors $\beps$ have sub-Gaussian tails. That is, there exists some constants $c,h > 0$ such that $\Pr(|\epsilon_i| > x) < 2\exp(-hx^2), \,\forall x > c$. Furthermore, $\Z = [z_{ij}]$ has bounded entries, i.e. $|z_{ij}| < M,\, \forall i,j$.
\item[({\bf A6})] The minimum eigenvalue of $\bSigma_\A$ is bounded (i.e. $\Lambda_{\min}(\bSigma_\A) \geq \eta > 0$), and the sample size $n$, the dimension $d$, the number of non-zero parameters $q$, and the minimum non-zero coefficient $b_{\min} \equiv \min\{\,|b_{\lambda,1}|,\dots,|b_{\lambda,q}|\,\}$ are such that 
$$
\frac{\log(d)}{n} + \frac{q \log(q)}{n b^2_{\min}}  \rightarrow 0.
$$
\end{itemize}
Conditions ({\bf A3}) and ({\bf A4}) guarantee that the zero and non-zero elements of $\b_\lambda$ can be distinguished from one another. In particular, ({\bf A3}) requires that the inactive variables $\Z_{\A^c}$ cannot be too correlated with the residuals $\y-a_\lambda\x - \Z\b_\lambda$. By the stationary conditions of \eqref{eq:params}, we know that $\|\btau_{\A^c}\|_\infty \leq 1$; here ({\bf A3}) ensures enough separation in this inequality as the dimension grows. Similarly, ({\bf A4}) requires that correlation between the active and inactive features cannot grow too quickly. This is related to the  irrepresentable condition, given in \citet{ZY2006}, which can be written as $\|\Z_{\A^{*c}}^T\Z_{\A^*}(\Z_{\A^*}^T\Z_{\A^*})^{-1}\mathrm{sign}(\bbeta_{\A^*})\|_\infty< 1$, where $\A^* = \mathrm{supp}(\bbeta)$.

Conditions ({\bf A5}) and ({\bf A6}) are somewhat standard in high-dimensional statistics. The sub-Gaussian tails of $\beps$ and boundedness of $\Z$  allow $d$ to grow quickly, so long as $\log(d)/n\rightarrow 0$. We assume sub-Gaussian tails in ({\bf A5}) for convenience; we could assume e.g. polynomial tails, at the cost of a slower rate of convergence. The condition on $ q \log(q) /(n b^2_{\min} )$ ensures that the non-zero elements of $\b_\lambda$ are large enough to be detected in the estimate $\hat\b_\lambda^0$. The condition on $\Lambda_{\min}({\bSigma}_\A)$ ensures that $\b_{\lambda\A}$ is identifiable. 

\begin{proposition}\label{thm:lassodist}
Let $\hat \b_\lambda^0$ be as in \eqref{eq:lassocoef} and define $T_\lambda = \x^T(\y-\Z\hat \b_\lambda^0)/\sqrt{n}$. Assume conditions ({\bf A1})-({\bf A6}) hold. Then under $H_{0,\lambda}: a_\lambda=0$,
\begin{equation}\label{eq:lassodist}
\frac{T_\lambda}{\sigma_\epsilon\sqrt{\x^T( {\bf I}_n - {\bf P}_\A)\x/n}} \rightarrow_d N\left(0,  1\right).
\end{equation}
\end{proposition}

Note that when $\lambda$ is chosen large enough so that $\b_\lambda = {\bf 0}$, then $\A = \emptyset$, and the variance of $T_\lambda$ is approximately $\sigma^2_\epsilon$, as in simple linear regression. On the other hand, if $\lambda=0$, the variance of $T_\lambda$ is approximately $\sigma^2_\epsilon \x^T\left({\bf I}_n-\Z(\Z^T\Z)^{-1}\Z^T\right)\x/n$, the variance of the classical score statistic used to test $\alpha=0$ in the model \eqref{mymodel}.

Unfortunately, the variance formula given in \eqref{eq:lassodist} depends on the support set $\A=\supp(\b_\lambda)$ and the residual variance $\sigma^2_\epsilon$, which are in general unknown. Estimating the residual variance $\sigma^2_\epsilon$ is required in a number of other procedures, such as the covariance test \citep{lockhart2013significance}, and there are a few options available \citep[see e.g.][]{fan2012variance}.

In order to estimate $\A$, we propose two options:
\begin{enumerate}
\item Use the observed support $\hat \A = \supp(\hat \b_\lambda^0)$ as an estimate of $\A$:
\begin{equation}\label{asymptoticvar}
\widehat{\var}(T_\lambda) =\hat \sigma^2_\epsilon\x^T( {\bf I}_n - {\bf P}_{\hat \A})\x/n.
\end{equation} 
We call this the \emph{asymptotic variance estimate} since it relies on the property that $\hat \A = \A$ with high probability, asymptotically.
\item Replace $\x^T\left({\bf I}_n-\bP_{\A}\right)\x/n$  with the upper bound of 1:
\begin{equation}\label{conservativevar}
\widehat{\var}(T_\lambda) =\hat \sigma^2_\epsilon\x^T\x/n = \hat \sigma^2_\epsilon.
\end{equation}
We call this the \emph{conservative variance estimate}.
\end{enumerate} 

In Section~\ref{sec:sim} we show that the asymptotic variance generally works well in practice. Using the conservative variance estimate has an appealing interpretation in light of Proposition~\ref{prop:score}. The penalized score test \eqref{eq:teststat} tests the effect of a single feature $\x=\x_j$ in \eqref{mymodel}, adjusting for the other features $\Z = [\x_k, k \neq j]$ with lasso regression. When using the penalized score test for testing the effect of $\x = \x_j$ for each $j =1\dots,d$ in turn, the conservative variance estimate will be the same for each $j=1,\dots,d$. Thus, the sparsity pattern of lasso regression, which results from comparing each $T_\lambda$ to $\sqrt{n}\lambda$, is the same as the set of rejections that results from applying the penalized score test to each feature in turn, using the same (conservative) significance threshold.

\subsection{Bias and the irrepresentable condition}\label{sec:bias}
As we saw in Section~\ref{whatnull}, when $\lambda > 0$  the penalized score test does not test the null hypothesis $H_0: \alpha=0$, as in multiple linear regression, but instead adopts the null hypothesis $H_{0,\lambda}: a_\lambda =0$.  Thus, if the penalized score test is used as a surrogate for classical tests of $H_{0} : \alpha = 0$ versus $H_1: \alpha \neq 0$, it may not be an unbiased test. In this section we investigate the {relationship between the penalized score test and unbiased tests of $H_{0} : \alpha = 0$, and show that, in the case of the lasso penalty, differences between the tests are closely related to the irrepresentable condition established by \citet{ZY2006}.
}

As discussed in Section~\ref{whatnull}, our main interest in this paper is in the interpretation and behavior of the penalized score test at a particular tuning parameter $\lambda$. However, in this section, we consider the behavior of the lasso-penalized score test as $\lambda \rightarrow 0$ in order to establish a connection with variable selection consistency results. We note that when using the lasso-penalized score test in practice, variable selection consistency is not necessary in order to obtain a valid test of $H_{0,\lambda}$.

First, we show that the lasso-penalized score statistic \eqref{eq:teststat} can be interpreted as a shifted version of a classical score statistic. Let  $\hat \A = \mathrm{supp}(\hat \b^0_\lambda)$ in \eqref{eq:lassocoef}, and consider testing for the effect of $\x$, adjusted for $\Z_{\hat \A}$, using the classical score test \eqref{classic}. In this case, the classical score statistic \eqref{classic} takes the form
\begin{equation}\label{classicA}
T_{\hat \A}  = \x^T\left ({\bf I}_n -{\bf P}_{\hat \A}\right)\y/\sqrt{n}.
\end{equation}
Suppose $\hat \b_\lambda^0$ is ordered such that  $\hat \b_{\lambda}^0 = (\hat b_{\lambda,1}^{0},\dots,\hat b_{\lambda,|\hat \A|}^{0},0,\dots,0)^T = (\hat \b_{\lambda\hat \A}^0,{\bf 0}^T)^T$, where $\min\{|\hat b_{\lambda,1}^{0}|,\dots,|\hat b_{\lambda,|\hat\A|}^{0}|\} >0$. Let $\bsigma_{x\hat \A} =\Z_{\hat \A}^T\x/n$, $\bSigma_{\hat \A} = \Z_{\hat \A}^T\Z_{\hat \A}/n$, and $\hat  \btau_{\hat \A} = \mathrm{sign}(\hat \b_{\lambda\hat \A}^0)$. Using the identity $\Z\hat \b_\lambda^0 ={\bf P}_{\hat \A}\y - \lambda\Z_{\hat \A}\bSigma_{\hat \A}^{-1}\hat \btau_{\hat \A}$, given in Equation 21 of \citet{tibshirani2012degrees}, we get that
\begin{equation}\label{classicB}
T_\lambda= T_{\hat \A}  + \sqrt{n}\lambda \bsigma_{x\hat \A}^T\bSigma_{\hat \A}^{-1}\hat \btau_{\hat \A}.
\end{equation}
Thus, the penalized score statistic $T_\lambda$ differs from the classical score statistic $T_{\hat \A}$ by $\sqrt{n}\lambda \bsigma_{x\hat \A}^T\bSigma_{\hat \A}^{-1}\hat \btau_{\hat \A}$. 

Now, suppose that we wish to test $H_0: \alpha=0$. In Section~\ref{whatnull}, we saw that $T_\lambda$ is centered around zero when $a_\lambda=0$. Therefore, unless $\alpha = a_\lambda$, the penalized score test may not be unbiased as a test of $H_0: \alpha = 0$. On the other hand, $T_{\hat \A}$ is centered around zero when $\hat \A$ contains all variables relevant to the outcome. A sufficient condition for $T_{\hat \A}$ to be centered around zero is then $\hat \A = \A^*$, where $\A^* = \mathrm{supp}(\bbeta)$. In order to make use of \eqref{classicB} to compare the penalized score test to an unbiased test of $H_0$, we thus consider the case where $(\y,\Z,\lambda)$ are such that $\mathrm{Pr}(\hat \A = \A^*)\rightarrow 1$ and $\Pr(\hat \btau_{\hat \A} = \mathrm{sign}(\bbeta_{\A^*}))\rightarrow 1$ under $H_0: \alpha=0$. \citet{ZY2006} showed that this holds provided that $\lambda \rightarrow 0$ and $\sqrt{n}\lambda\rightarrow \infty$ as $n\rightarrow \infty$, in addition to some assumptions on $(\y,\Z)$, which we omit here for ease of exposition. Note that we have not yet made any assumptions on the relationship between $\x$ and $\Z$. Under these assumptions, by \eqref{classicB} we have that 
\begin{equation}\label{classicC}
T_\lambda = T_{\A^*}+\sqrt{n}\lambda  \bsigma_{x\A^*}^T\bSigma_{\A^*}^{-1}\btau_{\A^*}+o_p(1),
\end{equation}
where $\btau_{\A^*} = \mathrm{sign}(\bbeta_{\A^*})$. Here, we can consider $T_{\A^*}$ to be the `oracle' test statistic: the score statistic for testing $H_0: \alpha=0$, which knows in advance the support of $\bbeta$.

Now, in order for the lasso to recover the support of $(\alpha,\bbeta)$ in lasso regression of $\y$ on $(\x,\Z)$, as in \eqref{lasso2}, \citet{ZY2006} showed that, in addition to conditions on $(\y,\Z,\lambda)$ required for \eqref{classicC} to hold, an \emph{irrepresentable condition} must hold. This condition implies (among other things) that $| \bsigma_{x\A^*}^T\bSigma_{ \A^*}^{-1} \btau_{\A^*}| <1$. Examining Proposition~\ref{prop:score}, we can see the connection between the irrepresentable condition and recovery of the support of $(\alpha,\bbeta)$. In the lasso solution \eqref{lasso2}, $\hat a_\lambda =0$ when $|T_\lambda| \leq\sqrt{n}\lambda$. Under $H_0:\alpha=0$, we have that $T_{\A^*} = O_p(1)$, and thus by \eqref{classicC}, $\Pr_{H_0}(|T_\lambda| \leq \sqrt{n}\lambda) \rightarrow 1$ when $| \bsigma_{x\A^*}^T\bSigma_{ \A^*}^{-1} \btau_{\A^*}| <1$. That is, when the irrepresentable condition is satisfied, the penalized score statistic $T_\lambda$ is close enough to the oracle statistic $T_{\A^*}$ for the decision rule used by the lasso (i.e. `reject $H_{0,\lambda}$ when $|T_\lambda| > \sqrt{n}\lambda$') to correctly identify that $\alpha=0$. 

In the preceding discussion we assumed that $\lambda$ was chosen in order to have $\hat \A = \A^*$.  However, this is not necessary in order to for the penalized score test to yield meaningful results. How far the penalized score statistic $T_\lambda$ deviates from zero under $H_0: \alpha=0$ depends more on the bias of $\hat \b^0_\lambda$ relative to $\bbeta$, rather than the support set $\hat \A$ per se. Choosing a smaller $\lambda$ will result in less bias in $\hat \b^0_\lambda$ relative to $\bbeta$, at the expense of a larger number of degrees of freedom spent in the lasso regression \eqref{eq:restricted}. We explore this issue numerically in Section~\ref{sec:sim}, and discuss it further in Section~\ref{choosinglambda}.

\subsection{Simulation study}\label{sec:sim}
In this section, we study the empirical behavior of the lasso-penalized score test.  We show that the test serves as a useful proxy for tests of $H_0: \alpha = 0$ provided $\lambda$ is small enough, and that it behaves like tests of marginal correlation when $\lambda$ is large.

First, we generated a matrix of correlated features $\X \in \mathbb{R}^{n \times d}$, where the rows were independently distributed $N_d(0, {\bf S})$ with ${\bf S}_{jk} = 0.5^{|j-k|}$. We then generated an outcome $\y \sim N_n(\X\bbeta, {\bf I}_n)$, where we set 10 elements of $\bbeta$ at random to be 0.4, and the rest to be zero. We used sample sizes of $n=50,75,100,150,200,300$ and $400$, and dimensions of $d=100$ and $d=300$. Results are averaged over $B=500$ simulated data sets, where $\bbeta$ was held constant over replications with the same dimension $d$.

We performed the lasso penalized score test on each feature in turn, for a sequence of values of $\lambda$. For the sake of comparison, we also performed simple linear regression of $\y$ on each feature, multiple linear regression of $\y$ on all features (for the cases where $d >n$), LDPE, and the `oracle' score test, which tests for the association of feature $\x_j$, and knows the support of the other features $\{\beta_k : k \neq j \}$. In order to make results comparable, we used the same estimate of the residual variance $\sigma^2_\beps$  in the penalized score test, multiple linear regression, and in simple linear regression, which we obtained using the refitted cross-validation method described by \citet{fan2012variance}. Note that we do not compare to methods of \citet{lee2013exact} and  \citet{taylor2014post}, which describe inference regarding only those features with non-zero coefficients in lasso regression, or the covariance test of \citet{lockhart2013significance}, which does not provide inference for individual features.

We declared a test to be significant when the resulting $p$-value was less than $1/d$. Since 10 feature are truly associated with the outcome, we would expect $(d-10)/d \approx 1$ false positives per simulated data set for an unbiased test, which we will refer to as the `nominal error rate'. For each test we calculated the expected false positives (EFP) and the power. For a particular test, if $p_{jk}$ is the $p$-value for feature $j$ on the $k^{th}$ simulated data set, EFP and power are given by 
\begin{align*}
\text{EFP} &= \frac{1}{B}\sum_{k=1}^B \sum_{j: \beta_j =0} 1\{p_{jk} < 1/d \} \\
\text{power} &= \frac{1}{B}\frac{1}{\|\bbeta\|_0} \sum_{k=1}^B \sum_{j : \beta_j \neq 0} 1\{p_{jk} < 1/d \}, 
\end{align*}
where $1\{\cdot\}$ is the indicator function and $\|\bbeta\|_0=10$ is the cardinality of $\bbeta$, where here we use the notation of Equation~\ref{linear1}. We chose to use a significance threshold in order to control EFP, but in principle one could use any method of error control, such as a {\^S}id{\'a}k, Bonferroni, or FDR correction.

In Figure~\ref{fig:score:slr_vs_mlr} we examine the $p$-values produced by the lasso-penalized score test when $\lambda = 0.005$ and $\lambda = 0.6$, for a single simulated data set with $d=100$ and $n=200$, and compare them to the $p$-values produced by multiple and simple linear regression. When $\lambda = 0.005$, 88 of the 100 coefficients are non-zero in the lasso regression on all features. Consequently, we see that the penalized score test behaves much like multiple linear regression including all 100 features. On the other hand, when $\lambda = 0.6$, only 3 features are included in the lasso regression on all features, and the $p$-values are similar to those from simple linear regression performed on each feature separately. For the sake of comparison, we also plot the multiple linear regression $p$-values against those from simple linear regression, which demonstrates that the behavior of these two tests are quite different. Note that the penalized score test $p$-values are not identical to those from classical tests: inference with the penalized score test is with respect to the parameter $a_\lambda$, given in \eqref{eq:parama}, which, as discussed in Section~\ref{whatnull}, measures different types of associations than those measured with simple or multiple linear regression.

\begin{figure}[htbp]
\centering
\spacingset{1}
\footnotesize
\includegraphics[width=\textwidth]{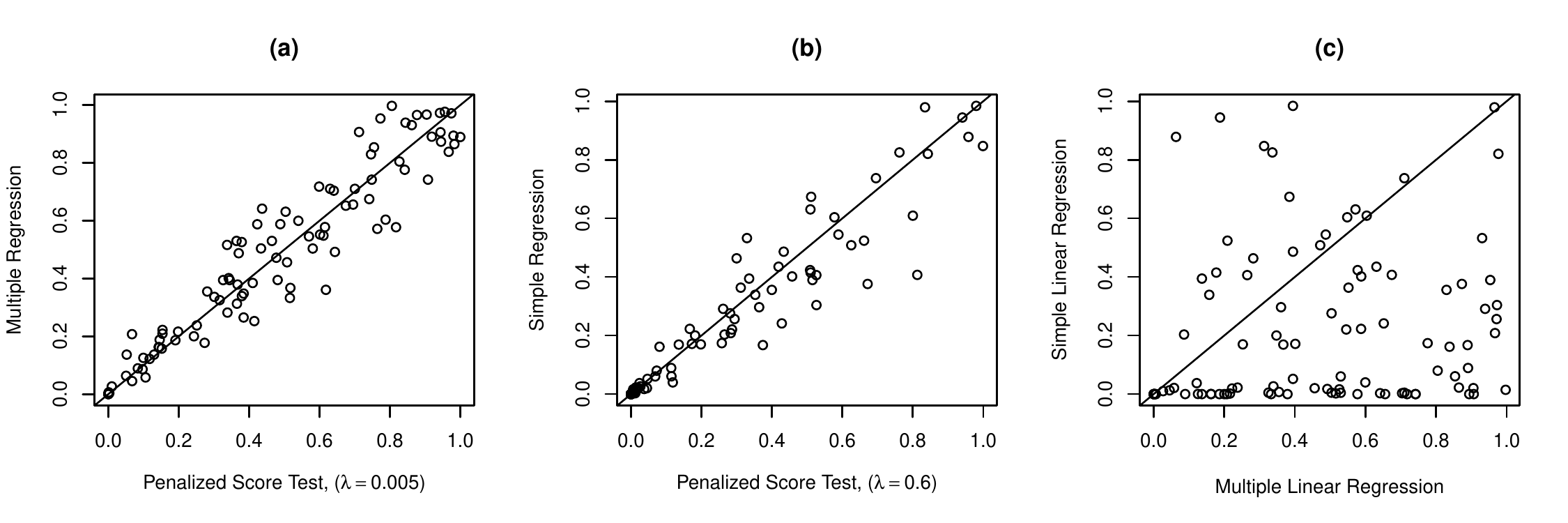}
\caption{Comparison of lasso-penalized score test $p$-values for $H_{0,\lambda}: a_\lambda = 0$ to traditional $p$-values for $H_0: \alpha = 0$, using the asymptotic variance formula \eqref{asymptoticvar}. (a) The $p$-values from multiple linear regression plotted against those from the penalized score test with $\lambda = 0.005$. (b) The $p$-values from simple linear regression plotted against those from the penalized score test with $\lambda = 0.6$. (c) Multiple and simple linear regression $p$-values plotted against each other.}
\label{fig:score:slr_vs_mlr}
\end{figure}

Figure~\ref{fig:score:sim} summarizes the results of the experiment over the $B=500$ replications. We see that simple linear regression provides high power, but also results in high EFP. Recall that simple linear regression will detect features which are marginally correlated with $\y$, whereas here we are interested in the conditional relationships. When $\lambda$ is large, the penalized score test has nearly identical power and EFP to simple linear regression; this is not surprising, since in that setting the penalized score test is almost identical to simple linear regression (see e.g. Figure~\ref{fig:score:slr_vs_mlr}). As $\lambda$ is decreased, the number of false positives converges to the nominal rate, for all sample sizes $n$ and dimensions $d$ considered, at the cost of lower power. This reduction in power should be expected: as $\lambda$ is decreased, more features enter the lasso regression \eqref{eq:lassocoef}, making it increasingly difficult to distinguish the effect of the feature $\x$ from the effects of the other correlated features in the model. In the extreme case of $\lambda =0$, or equivalently, using multiple linear regression, power is quite low. We see that the performance of our method using $\lambda = 0.05$ or $\lambda =0.07$ typically yields comparable type-I error, and slightly higher power, than LDPE.

\begin{figure}[htbp]
\centering
\spacingset{1}
\footnotesize
\includegraphics[width=\textwidth]{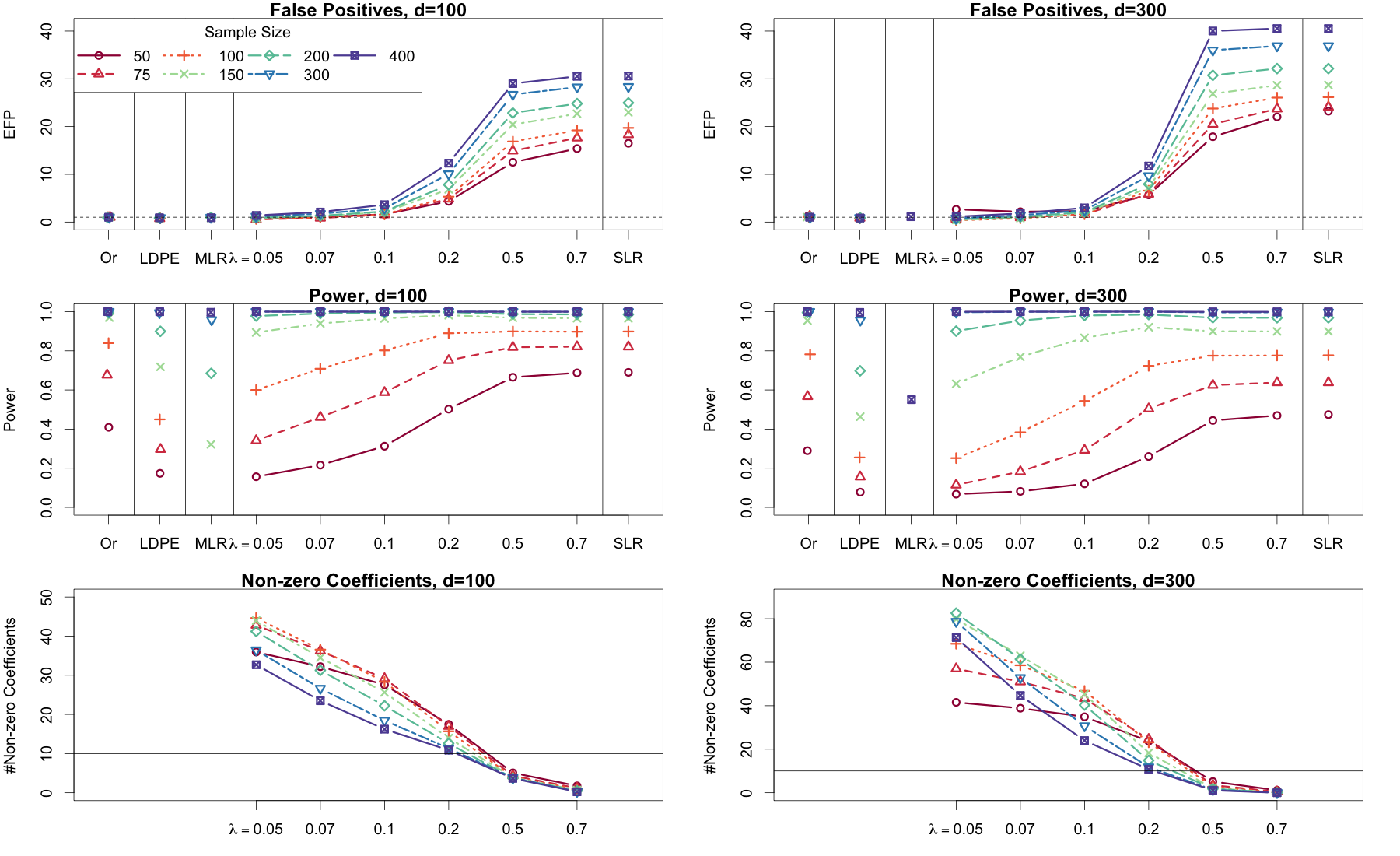}
\caption{Simulation experiments. `Or' indicates the the oracle test, `LDPE' indicates the method of \citet{zhang2011confidence} and \citet{van2013asymptotically}, `MLR' indicates multiple linear regression, and `SLR' indicates the results from simple linear regression. The six values of $\lambda$ correspond to the lasso-penalized score test. In the top panels, the horizontal line indicates the nominal error rate $(d-10)/d\approx 1$. In the bottom panels, the horizontal line indicates the true number of non-zero coefficients $\|\bbeta\|_0 = 10$. Results are averaged over 500 simulated data sets.}
\label{fig:score:sim}
\end{figure}

In the bottom panels of Figure~\ref{fig:score:sim}, we also plot the number of non-zero coefficients in lasso-penalized regression of $\y$ on $\X$. EFP is closest to the nominal rate when $\lambda$ is smallest, and here we see that this results in many more non-zero lasso coefficients than the number of truly non-zero coefficients. That is, in order to strictly control the error rate, it seems beneficial to choose $\lambda$ to be smaller than one would choose it if variable selection with the lasso were the goal. This assertion is supported by the theory in Section~\ref{sec:bias}, where we showed that $T_\lambda$ can diverge to $\pm \infty$ under $H_0: \alpha = 0$, if $\lambda$ is chosen so that the lasso selects the correct model (i.e. $\sqrt{n}\lambda \rightarrow \infty$ while $\lambda \rightarrow 0$). On the other hand, in Section~\ref{whatnull} we showed that $T_\lambda$ is centered around zero under $H_0: \alpha = 0$ when $\sqrt{n}\lambda \rightarrow 0$.

\subsection{Diabetes data}\label{sec:diabetes}
Here we re-examine the diabetes data set from Section~\ref{sec:1:intro}. We apply the lasso-penalized score test, for a range of 300 values of $\lambda$ between $0$ and $50$. We estimate $\sigma^2_\epsilon$ by the residual variance from multiple linear regression on all features.

\begin{figure}[htbp]
\centering
\spacingset{1}
\footnotesize
\includegraphics[width=\textwidth]{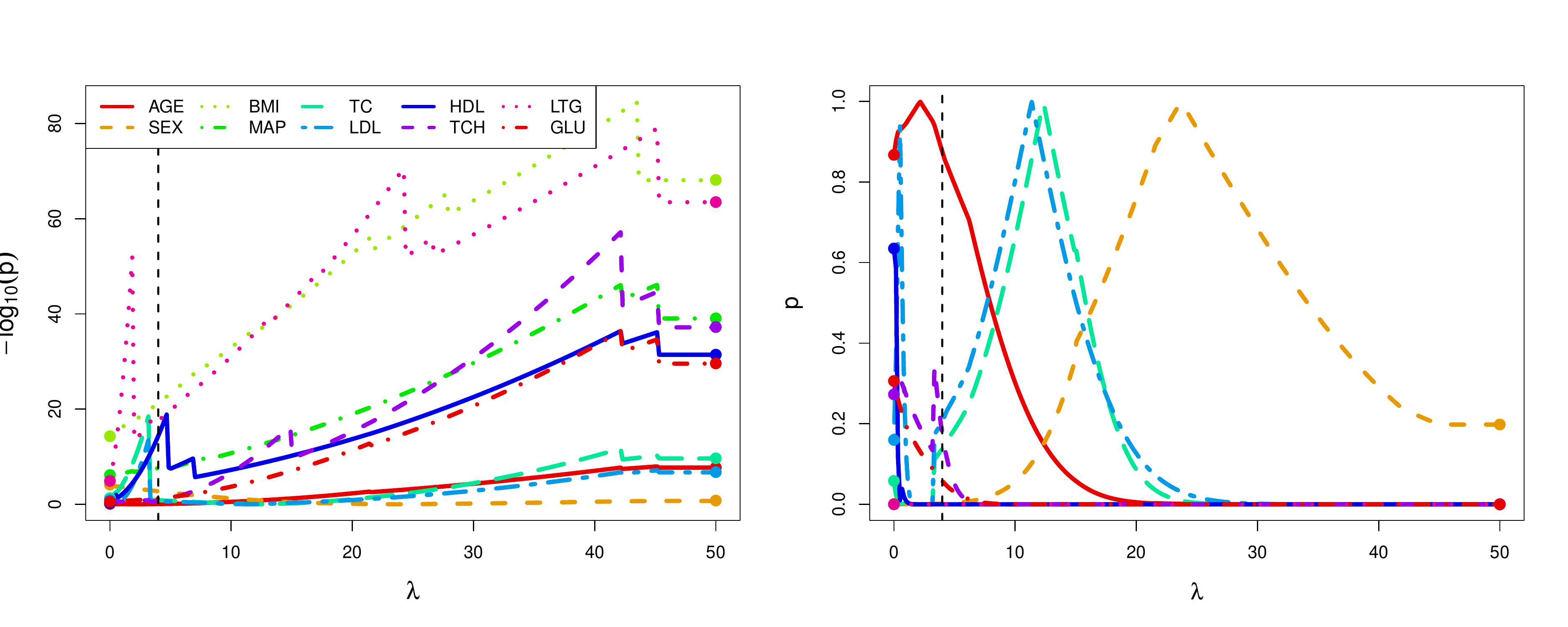}
\caption{Diabetes data set. Lasso-penalized score test $p$-values were generated using the asymptotic variance formula \eqref{asymptoticvar}. The vertical line at $\lambda=4$ indicates the value chosen to produce $p$-values for the penalized score test given in Table~\ref{tab:diab}. Dots at $\lambda=0$ and $\lambda=50$ indicate $p$-values from multiple linear regression on all features, and simple linear regression on each feature alone.    }
\label{fig:score:diab}
\end{figure}

Figure~\ref{fig:score:diab} summarizes the results of this analysis, showing both $-\log_{10}$ and un-transformed $p$-values for each of the $\lambda$ values, using the asymptotic variance formula \eqref{asymptoticvar}. For comparison, we also plot the multiple and simple linear regression $p$-values; these are displayed in Figure~\ref{fig:score:diab} at the far left side ($\lambda=0$) and on the far right side ($\lambda=50$) respectively. We see that the $p$-values from the penalized score test interpolate the multiple linear regression $p$-values and the simple linear regression $p$-values, and vary widely depending on the value of $\lambda$ chosen. The $p$-values for each feature are piece-wise continuous, with jumps when the set $\mathrm{supp}(\hat \b_\lambda^0) = \hat \A$ changes. The jumps typically result in smaller $p$-values immediately after an element of $\hat \b_\lambda^0$ becomes non-zero, since the variance of the test statistic, $\sigma^2_\epsilon\x^T\left({\bf I}_n-\bP_{\hat \A}\right)\x/n$, will be smaller when the size of the set $\hat \A$ is larger. This phenomenon is investigated in greater detail in Section~\ref{sec:thresh}.

The vertical line in Figure~\ref{fig:score:diab} indicates the value of $\lambda=4$, chosen to produce the $p$-values in Table~\ref{tab:diab}. With this value of $\lambda$, 4 of the 10 covariates (AGE, TC, LDL, and TCH) have zero coefficients in lasso regression on all features, while the rest are non-zero. Lasso regression on all features with this choice of $\lambda$ yields an $R^2$ of 0.50, compared to the $R^2$ of 0.52 using multiple linear regression on all features. Both in Table~\ref{tab:diab} and in Figure~\ref{fig:score:diab}, we see that this value of $\lambda$ results in $p$-values which are qualitatively similar to those from multiple linear regression on all features. However, it is notable that HDL appears strongly associated in this lasso-penalized score test, but not in the multiple linear regression. The multiple linear regression results suggest that the effects of HDL can be explained by other features in the data set; using the penalized score test with $\lambda=4$, the effects of other features are sufficiently shrunken towards zero so that HDL appears associated with the response.

\begin{figure}[htbp]
\centering
\spacingset{1}
\footnotesize
\includegraphics[width=\textwidth]{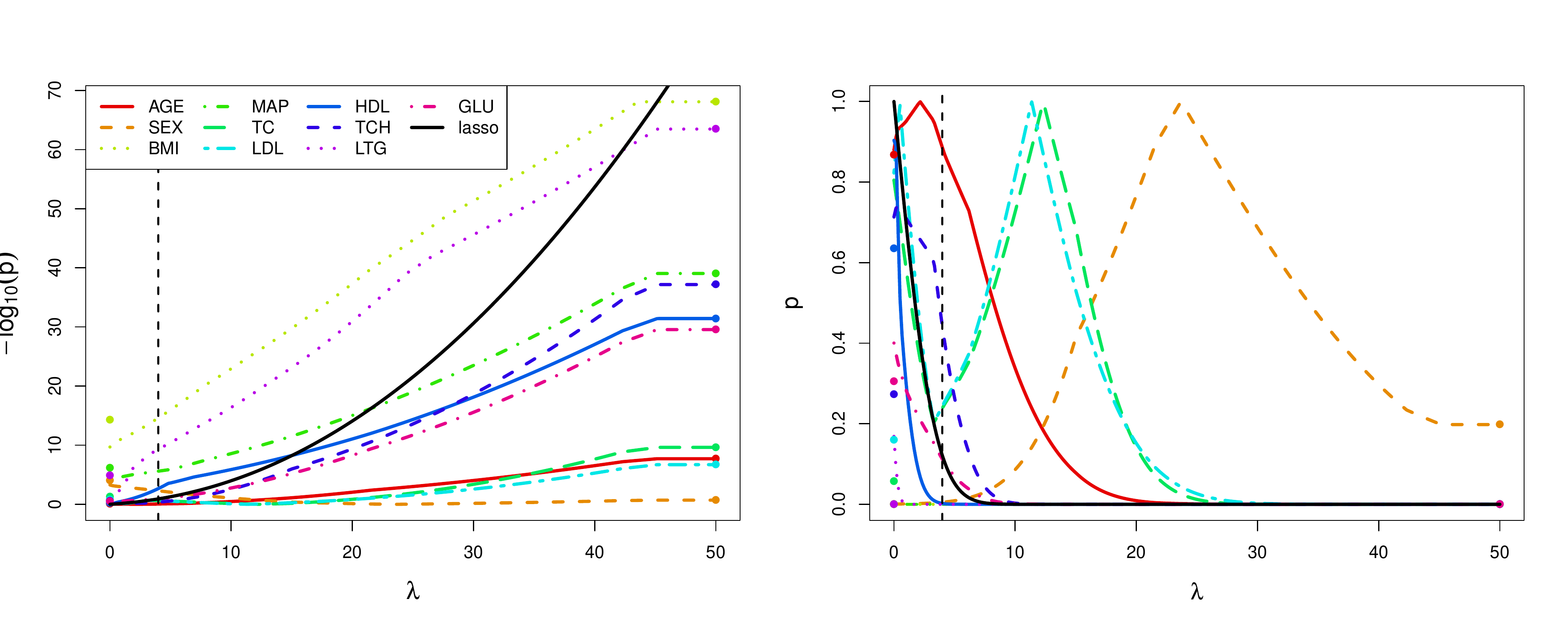}
\caption{Diabetes data set. Lasso-penalized score test $p$-values generated using the conservative variance formula \eqref{conservativevar}. The lasso decision rule, shown in black, corresponds to $\Phi(-2\sqrt{n}\lambda/\sigma_\epsilon)$, where $\Phi(\cdot)$ is the standard normal distribution function.}
\label{fig:score:diab2}
\end{figure}

In Figure~\ref{fig:score:diab2} we once again show $p$-values corresponding to the lasso score test applied to the diabetes data. This time, however, we calculate $p$-values using the conservative variance estimate $\widehat{\var}(T_\lambda) = \sigma^2_\epsilon$, in \eqref{conservativevar}, for each feature. Since this variance formula does not depend on the support of $\hat \b^0_\lambda$, the $p$-values are continuous curves. Further, since the same reference distribution is used for each feature, the decision rule which yields the sparsity pattern of the lasso, `reject $H_{0,\lambda}$ when $|T_\lambda| > \sqrt{n}\lambda$', corresponds to the same $p$-value threshold for each feature. This threshold is displayed as a thick black line in Figure~\ref{fig:score:diab2}; when the $p$-value for a feature crosses the line, its coefficient becomes non-zero in the lasso regression. For instance, at $\lambda = 4$, the $p$-value threshold is 0.12. The variables AGE, TC, LDL, and TCH have $p$-values above this threshold and thus have zero coefficients.

\subsection{Assessing the impact of thresholding}\label{sec:thresh}
It is well-known that the finite sample distributions of estimators that involve thresholding may be far from their large-sample limits. The canonical example of this phenomenon is Hodges' `super-efficient' estimator \citep[see e.g.][page 589]{LC_hodges}; similar behavior has also been observed in lasso-type estimators \citep{KF2000,potscher2009distribution}. In this section, we explore how thresholding can impact the type-I error rate of the penalized score test, when the lasso is used as the penalty. We follow the example of \citet{leeb2005model}, and consider the two-variable case, where the exact distribution of $T_\lambda$ is simple to obtain. As we will see, our proposed test can be either conservative or anti-conservative, depending on the underlying parameters, and the nominal type-I error of the test.

Suppose we are interested in the effect of a variable $\x\in \mathbb{R}^n$, adjusted for an additional variable $\z \in \mathbb{R}^n$.   Further suppose that $\y = \alpha\x + \beta\z + \beps$, where $\x$ and $\z$ are fixed,  $\x^T\z/n = \rho$, $\x^T\x/n=\z^T\z/n=1$, and $\beps \sim  N_n(0, {\bf I}_n)$. As a reminder, we are testing the effect $a_\lambda = \alpha + \rho(\beta - b_\lambda)$, given in \eqref{eq:parama}.

 In this simple case, the lasso-penalized score test has two steps:
\begin{enumerate}
\item Regress $\y$ on $\z$ using the lasso. This corresponds to soft-thresholding the quantity $\y^T\z/n$. That is, we set $\hat b^0_\lambda =\mathrm{sign}(\y^T\z) (|\y^T\z/n| - \lambda)_+$, an estimate of $b_\lambda$ under the null hypothesis $H_{0,\lambda}: a_\lambda=0$.\label{step1}
\item Construct $T_\lambda = \x^T(\y -\hat b_\lambda^0\z) /\sqrt{n}$, and compare to a normal reference distribution, with variance to be specified next.
\end{enumerate}
Under $H_{0,\lambda}: a_\lambda = 0$, Proposition~\ref{thm:lassodist} shows that in large samples, $T_\lambda$ should have variance $1-\rho^2$ when $|\E\y^T\z/n| > \lambda$, and 1  otherwise. In finite samples, we can either use the asymptotic estimate \eqref{asymptoticvar} (i.e. use $\widehat{\var}(T_\lambda) =1-\rho^2$ when $|\y^T\z/n| > \lambda$, and $\widehat{\var}(T_\lambda)=1$  otherwise), or the conservative estimate \eqref{conservativevar} (i.e. always use variance  $\widehat{\var}(T_\lambda)=1$). We will investigate the behavior of the test for both estimators.

Here, $T_\lambda$ can be written explicitly as
\begin{equation} \label{texplicit}
T_\lambda =\begin{cases}
\sqrt{n}[(1-\rho^2)\alpha + \rho\lambda] + (\x - \rho\z)^T\beps/\sqrt{n} & \,\,\text{if}\quad \z^T\y/n \geq \lambda\\
\sqrt{n}(\alpha + \rho\beta)+ \x^T\beps/\sqrt{n} & \,\,\text{if}\quad|\z^T\y/n| < \lambda \\
\sqrt{n}[(1-\rho^2)\alpha - \rho\lambda] + (\x - \rho\z)^T\beps/\sqrt{n} & \,\,\text{if}\quad \z^T\y/n \leq -\lambda
\end{cases} .
\end{equation}

It is easy to see that, conditioned on $\z^T\y \sim N(n(\rho\alpha+\beta), n)$, $T_\lambda$ is normally distributed. To find the marginal distribution of $T_\lambda$, we simply calculate $\E_{\z^T\y }[\,T_\lambda \mid \z^T\y \,]$, by numerical integration.

Our goal is to determine the impact of thresholding on type-I error. In order to do this, we choose $(\alpha,\beta)$ such that (i) the null hypothesis $H_{0,\lambda}: a_\lambda = 0$ is true, and (ii) the probability $\mathrm{Pr}(\z^T\y/n \geq \lambda)$ is controlled to be $\gamma$.  Since $\z^T\y/n \sim N(\rho\alpha+\beta, 1/n)$, we must have that $\beta = \Phi^{-1}(\gamma)/\sqrt{n} + \lambda - \rho\alpha$ in order to achieve $\mathrm{Pr}(\z^T\y/n \geq \lambda)=\gamma$, where $\Phi(\cdot)$ is the standard normal distribution function. In order to have $a_\lambda = 0$, we must have $\alpha+\rho(\beta-b_\lambda)=0$, where $b_\lambda =\mathrm{sign}(\E[\,\y^T\z\,]) (|\E[\,\y^T\z/n\,]| -\lambda)_+ = \mathrm{sign}(\alpha\rho + \beta)(\alpha\rho + \beta-\lambda)_+$. Thus, with restrictions (i) and (ii), we must have that 
\begin{align*}
\alpha &= \begin{cases}
-\rho \lambda/(1-\rho^2) & \text{if}\quad \gamma > 0.5   \\
-\rho (\Phi^{-1}(\gamma)/\sqrt{n} + \lambda)/(1-\rho^2) & \text{if}\quad \Phi(-2\sqrt{n}\lambda) \leq \gamma \leq 0.5\\
\rho \lambda/(1-\rho^2) & \text{if}\quad \gamma <  \Phi(-2\sqrt{n}\lambda) \\
\end{cases} \\
 \beta &=  \Phi^{-1}(\gamma)/\sqrt{n} + \lambda - \rho\alpha.
\end{align*}
The cases $\gamma > 0.5$, $\Phi(-2\sqrt{n}\lambda) \leq \gamma \leq 0.5$, and $\gamma <  \Phi(-2\sqrt{n}\lambda)$ correspond to the cases $\E\z^T\y/n > \lambda$, $|\E\z^T\y/n| \leq \lambda$, and $\E\z^T\y/n < -\lambda$ respectively. Note that for fixed $\lambda$, $\Phi(-2\sqrt{n}\lambda)\approx 0$ for large $n$. 

When $|\E\z^T\y/n| \leq \lambda$, or equivalently, when $ \Phi(-2\sqrt{n}\lambda) \leq \gamma \leq 0.5$, then $ b_\lambda = 0$ in truth. However, with probability $\gamma$, we will erroneously have $\hat{ b}^0_{\lambda} > 0$. Likewise, when $\E\z^T\y/n > \lambda$, or equivalently, when $\gamma > 0.5$, we then have $b_\lambda > 0$, but with probability $1-\gamma-\Phi[-2\sqrt{n}\lambda - \Phi^{-1}(\gamma)]$, which is approximately $1-\gamma$ for $\sqrt{n}\lambda$ large enough, we erroneously set $\hat b_\lambda^0 =0$. Thus, by varying $\gamma$, we can examine the impact of erroneously including or excluding a feature in the lasso regression. 

Figure~\ref{fig:local} displays the relative type-I error of the test under $H_{0,\lambda}: a_\lambda = 0$, i.e. (observed type-I error)/(nominal type-I error), for a range of values of $\gamma$, using both the asymptotically derived variance estimate \eqref{asymptoticvar} and the conservative variance estimate \eqref{conservativevar}. We chose $\lambda=0.2$, and $n=500$. Note that by parametrizing the coefficients by $\gamma$, the results depend only very weakly on $n$; similar curves can be obtained for arbitrarily large sample sizes. 

\begin{figure}[htbp]
\centering
\spacingset{1}
\footnotesize
\includegraphics[width=\textwidth]{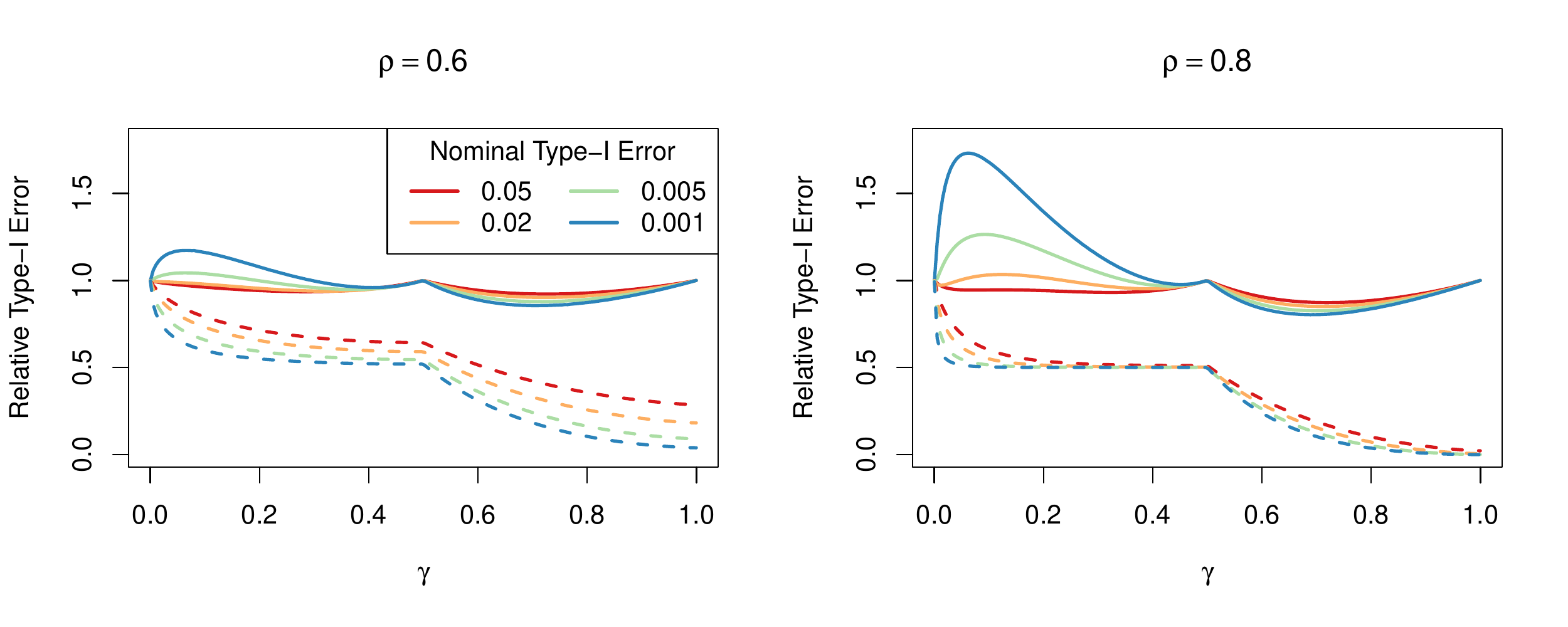}
\caption{Relative type-I error rate, i.e. (observed type-I error)/(nominal type-I error), as a function of $\gamma = \Pr(\hat b^0_\lambda > 0)$. Solid lines indicate the error rates for the asymptotic variance formula \eqref{asymptoticvar}, while dashed lines indicate the error rates when using the conservative variance \eqref{conservativevar}. Note that the inflection point at $\gamma=0.5$ is due to the fact that $b_\lambda =0$ when $\Phi^{-1}(-2\sqrt{n}\lambda)<\gamma \leq 0.5$, while $b_\lambda > 0$ when $\gamma > 0.5$.}
\label{fig:local}
\end{figure}

We see that when we use the conservative variance (dashed lines), the test is indeed conservative. When $\gamma=0$ the observed error rate is identical to the nominal rate, while the test becomes increasingly conservative as $\gamma$ increases. 

On the other hand, when we use the asymptotic variance formula (solid lines), the test can be anti-conservative when both the nominal type-I error rate and $\gamma$ are small, but is otherwise conservative. In general, the behavior of the test is worse for small type-I error rates, and when the correlation between the features is larger.

{A reviewer suggested that one could estimate the distribution of $T_\lambda$ more accurately using the methods described by  \citet{andrews2009hybrid}. In order to implement this procedure, one would use a critical value for $T_\lambda$ based on a hybrid of $m$-of-$n$ bootstrap samples, and the critical values based on an asymptotic distribution of the test.  Alternately, the framework of \citet{berk2012valid} could be used to obtain a conservative version of the test.} 

\section{The ridge-penalized score test}\label{sec:ridge}
In this section, we examine in greater detail the penalized score test where the ridge penalty is used in \eqref{eq:restricted}. That is, we first obtain the penalized coefficient vector
\begin{equation}
\hat \b^0_\lambda = \argmin_{\b \in \mathbb{R}^{d-1}}\left\{\|\y - \Z\b\|_2^2/(2n) + \lambda \|\b\|^2_2/2\right\},
\end{equation}
and then form the test statistic $T_\lambda = (\y-\Z\hat \b_\lambda^0)/\sqrt{n}$ in \eqref{eq:teststat}. 

First, we prove Proposition~\ref{prop:mixed}. Let ${\H}_{\Z} \equiv \Z(\lambda {\bf I}_{d-1} + \Z^T\Z/n)^{-1}\Z^T/n$ denote the $n\times n$ smoother matrix from ridge regression. Here we can write $T_\lambda$ as
\begin{equation}\label{eq:ridge}
T_\lambda = \x^T({\bf I}_n -{\H}_{\Z})\y/\sqrt{n}.
\end{equation}

We now show that \eqref{eq:ridge} can be interpreted as the score statistic from a mixed model. If we assume  \eqref{eq:mixed}, then the marginal distribution of $\y$ is
\begin{equation}\label{ymarg}
\y \sim N_n\left(\alpha \x, \, \sigma^2_\epsilon [(n\lambda)^{-1}\Z\Z^T + {\bf I}_n]\right).
\end{equation}
Writing $l(\alpha)$ as the log-likelihood of the data under \eqref{ymarg}, we can write the score, evaluated at $\alpha=0$, as
$$ \dot l(0) = \frac{1}{\sigma^2_\epsilon} \x^T\left[(n\lambda)^{-1}\Z\Z^T + {\bf I}_n\right]^{-1}\y.$$
Recognizing that $[(n\lambda)^{-1}\Z\Z^T + {\bf I}_n]^{-1} = ({\bf I}_n - \H_{\Z})$, we see that the score for testing $\alpha=0$ in \eqref{eq:mixed} is
$$\dot l(0) =\frac{1}{\sigma^2_\epsilon} \x^T({\bf I}_n -\H_{\Z})\y, $$
which is a scaled form of~\eqref{eq:ridge}. That is, the ridge-penalized score statistic is equivalent to the score statistic for testing the effect of feature $\x$ in a mixed model, where all the other features have normally distributed random effects with variance $\sigma^2_\epsilon(n\lambda)^{-1}$. 

The distribution of the ridge-penalized score statistic $T_\lambda$, or equivalently $\dot l(0)$, depends on whether we consider $\bbeta$ to be fixed under the null hypothesis $H_{0,\lambda}:a_\lambda = 0$ in \eqref{eq:params}, or random under the null hypothesis $H_0:\alpha = 0$ in \eqref{eq:mixed}.  With $\bbeta$ fixed, solving \eqref{eq:params} when $a_\lambda = 0$ yields $\Z\b_\lambda = {\bf H}_\Z(\alpha\x + \Z\bbeta)$, and thus $\E[T_\lambda \mid \bbeta ]=\x^T({\bf I}_n - {\bf H}_\Z)(\alpha\x + \Z\beta)/\sqrt{n} = \sqrt{n}[\,\alpha + \bsigma_{xz}^T(\bbeta - \b_\lambda)\,] = 0$. Thus, we can write the exact distribution of $T_\lambda$ under $H_{0,\lambda}: a_\lambda = 0$ as
\begin{equation}\label{conditional}
T_\lambda \mid \bbeta \stackrel{H_{0,\lambda}}{\sim} N(0, \sigma^2_\epsilon \x^T({\bf I}_n - {\bf H}_\Z)^2\x/n),
\end{equation}
since $T_\lambda$ is simply a linear function of a normal vector with mean zero. On the other hand, when $H_0: \alpha = 0$ in the mixed model \eqref{eq:mixed}, we can use the marginal distribution of $\y$ in \eqref{ymarg} to obtain
\begin{equation}\label{unconditional}
T_\lambda \stackrel{H_0}{\sim} N(0, \sigma^2_\epsilon \x^T({\bf I}_n - {\bf H}_\Z)\x/n).
\end{equation}
It is easy to show that $\x^T({\bf I}_n - {\bf H}_\Z)^2\x \leq  \x^T({\bf I}_n - {\bf H}_\Z)\x$, and thus $T_\lambda$ has a smaller variance if we assume that $\bbeta$ is fixed.

In our usual interpretation of the penalized regression~\eqref{eq:restricted}, we do not consider the effects of the features to be random draws from some population. Instead, we can motivate the use of ridge regression with the desire for an estimate of $\bbeta$ with smaller variance than the multiple linear regression estimate \citep{draper1979ridge}. Using a mixed-effects framework when $\bbeta$ is non-random is also discussed by \citet{hodges2010adding} in the context of spatial statistics, and \citet{greenland2000should} in the context of epidemiology.

\section{Extension to other sparsity-inducing penalties}\label{sec:noncon}
We have shown that the sparsity pattern of the lasso can be interpreted as resulting from a statistical test. In this section, we show that other sparsity-inducing penalties, such as \emph{smoothly clipped absolute deviation} (SCAD) \citep{fan2001variable} and the \emph{elastic net} \citep{ZH2004}, can also be interpreted in a similar framework. 

Suppose we obtain sparse estimates $(\hat \alpha_\lambda,\hat\b_\lambda)$ of the linear regression parameters $(\alpha,\bbeta)$, by solving the optimization problem
\begin{equation}\label{ext1}
(\hat a_\lambda,\hat \b_\lambda) =\argmin_{(a,\b) \in \mathbb{R}^d} \left\{\frac{1}{2n}\|\y-a\x-\Z\b\|_2^2 + \lambda J(a,\b)\right\}.
\end{equation}
A common characteristic of sparsity-inducing penalties is non-differentiability of $J(a,\b)$ around zero. Suppose $J(a,\b)$ is symmetric about zero, with subdifferential $\frac{\partial}{\partial a} J(a,\b) \mid_{a=0}= [-1,1]$. The elastic net and SCAD are two such examples. Now, denote 
\begin{equation}\label{ext2}
 \hat \b^0_\lambda = \argmin_{\b \in \mathbb{R}^{d-1}}\left\{\frac{1}{2n}\|\y-\Z\b\|_2^2 + \lambda J(0,\b)\right\}.
\end{equation}
A necessary condition for $\hat a_\lambda = 0$ in \eqref{ext1} is that $\x^T(\y-\Z\hat\b_\lambda^0)/n \in [-\lambda,\lambda]$. Using the notation from Section~\ref{sec:method}, we have that
$$\{ \hat a_\lambda = 0\} \implies \{|T_\lambda| \leq \sqrt{n}\lambda\}.$$
In addition, if $J(a,\b)$ is convex, as is the case for the elastic net, then
$$\{\hat a_\lambda = 0 \}\iff \{|T_\lambda| \leq \sqrt{n}\lambda\}.$$

In other words, the sparsity pattern of coefficient vectors induced by any convex penalty with subdifferential $[-1,1]$ at the origin can be interpreted as a decision made based on the penalized score statistic $T_\lambda$. 

For non-convex penalties, such as SCAD, the penalized score test gives only a necessary condition for regression parameters to be zero. In practice, however, the penalized score test is in some sense both necessary and sufficient to determine the sparsity pattern produced by non-convex penalties. Solutions to non-convex problems like SCAD are often found using coordinate-descent procedures, which solve for a local optimum of \eqref{ext1} by iteratively minimizing with respect to each element of $(a,\b)$ \citep{zou2008one,breheny2011coordinate,mazumder2011sparsenet}. If we use $(0,\hat \b^0_\lambda)$ as initial values, and then solve $\eqref{ext1}$ using coordinate descent, the algorithm will converge to   $(0,\hat \b^0_\lambda)$ when $|T_\lambda| \leq \sqrt{n}\lambda$. That is, $\{|T_\lambda| \leq \sqrt{n}\lambda\}$ is also sufficient for $\{\hat a_\lambda=0\}$, when using this algorithm.

In order to obtain $p$-values for these other sparsity-inducing penalties, we require an appropriate reference distribution, which we leave for future work.

\section{Discussion} \label{choosinglambda}
In this paper, we presented the penalized score test, in which the hypothesis being tested depends on the value of the tuning parameter $\lambda$. Therefore, {$\lambda$ should be chosen to yield a test which is scientifically meaningful. For instance, if the simple linear regression parameter $\alpha + \bsigma_{xz}^T\bbeta$ is a scientifically meaningful target of inference, one should choose $\lambda$ to be large (i.e. perform simple linear regression).} On the other hand, $\lambda$ should be chosen as small as possible when the multiple linear regression parameter $\alpha$ in \eqref{mymodel} is of interest. Perhaps the simplest way to choose $\lambda$ in this case is to specify how many degrees of freedom we are willing to invest in estimating the nuisance parameter $\bbeta$. {Whether this controls type-I error of the penalized score test at an acceptable level is highly context-specific, and in general difficult to ascertain.} With any $\lambda > 0$, some bias in $\hat \b^0_\lambda$ relative to $\bbeta$ will be incurred, in which case the penalized score test may be thought of as a pragmatic approximation to classical tests, {when multiple linear regression is undefined, or produces coefficient estimates which are too variable to be useful.} 

In this manuscript, we focused on testing the hypothesis $H_{0,\lambda}: a_\lambda=0$ using a score test. The test does estimate effect sizes or provide confidence intervals. However, estimates of effect size can be obtained by fitting the sample version of \eqref{eq:params}, i.e. $(\hat a_\lambda,\hat \b_\lambda) = \argmin_{(a,\b)}\left\{\|\y-a\x-\Z\bbeta\|_2^2/(2n) + \lambda J(\b)\right\}$. Confidence intervals for $a_\lambda$ when using the ridge penalty ($J(\b) = \|\b\|_2^2/2$) are available from mixed-model theory. For the lasso ($J(\b) = \|\b\|_1$), slight modifications of our theory in Section~\ref{sec:lasso} can be used to show that 
$
\sqrt{n}(\hat a_\lambda - a_\lambda) \rightarrow_d N\left(0, \sigma^2_\epsilon\left[ \x^T({\bf I}_n  - {\bf P}_\A)\x/n\right]^{-1} \right)
,$ under ({\bf A1-6}). This result might be used for a penalized version of the Wald test.

Several possible extensions of the proposed method are outlined here. Instead of testing for the effect of a single feature $\x \in \mathbb{R}^n$, we can test for groups of $k$ variables $\X \in \mathbb{R}^{n \times k}$, with the score statistic
$
T_\lambda = \X^T(\y - \Z\hat\b_\lambda^0)/\sqrt{n} \in \mathbb{R}^k
$. The distribution of $T_\lambda$, under an appropriate null hypothesis, follows in a straightforward way from Proposition~\ref{prop:mixed} for a ridge penalty, or from Proposition~\ref{thm:lassodist}, for a lasso penalty. From Proposition~\ref{prop:score}, we know that in the lasso regression of $\y$ on $(\X,\Z)$, one or more of the coefficients associated with $\X$ is non-zero if and only if $\|T_\lambda\|_\infty > \sqrt{n}\lambda$, where $T_\lambda$ is constructed using the lasso penalty. Analogous to Proposition~\ref{prop:score}, the sparsity pattern of the group lasso \citep{grouplasso} and the standardized group lasso \citep{standgl} could also be understood in terms of restrictions on this score statistic $T_\lambda$, for an appropriate choice of penalty function $J(\b)$.

The lasso-penalized score test is implemented in the \verb!lassoscore! \verb!R! package, available on the Comprehensive R Archive Network (CRAN).

\section*{Acknowledgments} We thank Dezure Ruben and Peter B\"uhlmann for providing \verb!R! code for the LDPE method, and Adam Szpiro and Ken Rice for valuable insights.


\appendix
\section{Appendix: Technical Proofs}
In this section we prove Proposition~\ref{thm:lassodist}. First, in Lemma~\ref{basiclemma}, we state and prove a basic result. We then proceed by stating and proving additional lemmas needed in the proof of Proposition~\ref{thm:lassodist}.

\setcounter{lemma}{0}

\begin{lemma}\label{basiclemma}
Let $\{X_{ij}:\, i = 1, \dots, n;  j =1,\dots,d \}$ be a set of random variables such that $\{X_{1j}, \dots, X_{nj}\}$ are mutually independent. Assume $\E(X_{ij}) = 0$, and that there exist $h,c > 0$, not depending on $i$ or $j$ such that  $\Pr(|X_{ij}| \geq x) \leq 2 \exp(-h x^2), \,\forall x >c$. Denote $Z_{j} = \sum_{i=1}^n X_{ij}/\sqrt{n}$. Then 
$$\max_{j =1,\dots,d}|Z_{j}| = O_p\left(\log^{1/2}(d) \right).$$
\end{lemma}

\begin{proof}
First we state a well-known equivalent definition of a sub-Gaussian random variable, which follows from, e.g. Lemma~14.2 in \citet{buhlmann2011statistics}. We have that $\Pr(|X_{ij}| \geq x) \leq 2 \exp(-h x^2)$ for $ x > c$ if and only if $M_{X_{ij}}(t) \leq \exp(-k t^2)$ for some $k>0$, where $M_{X_{ij}}(t)$ is the moment generating function of $X_{ij}$. Using this fact, we know that $Z_{j}$ is sub-Gaussian since $M_{Z_{j}}(t) = \prod_{i=1}^nM_{X_{ij}}(t/\sqrt{n}) \leq \exp(-kt^2)$. Applying the union bound, we get that
\begin{align*}
\Pr\left[\max_{j =1,\dots,d}|Z_{j}| > t\log^{1/2}(d)\right] &\leq \sum_{j=1}^d \Pr\left[\, |Z_{j}| > t\log^{1/2}(d)\right] \\&\leq  2d\exp(-h t^2\log(d)) \\
& = 2\exp(\log(d)[1-ht^2]) \\
& \leq 2\exp(\log(2)[1-ht^2]),
\end{align*}
where the last inequality holds when $ht^2 > 1$ and $d \geq 2$. Thus, we can choose a large value of $t$, not depending on $d$, such that $\Pr\left[\max_{j =1,\dots,d}|Z_{j}| > t\log^{1/2}(d)\right] $ is arbitrarily small, which gives the result.
\end{proof}

\begin{lemma}\label{1:mainlemma}
Suppose conditions ({\bf A1}),({\bf A5}) and ({\bf A6}) hold. Then the estimator 
$$
\tilde \b = \argmin_{\b : \|\b_{\A^c}\| = 0 }\left\{ \frac{1}{2n}\|\y-a_\lambda\x-\Z\b\|_2^2 + \lambda \|\b\|_1\right\} 
$$
satisfies $\|\tilde \b_\A - \b_{\lambda\A}\|_2 = O_p\left(\sqrt{q\log (q)/n }\right)$.
\end{lemma}

\begin{proof}
To simplify the notation, assume without loss of generality that $a_\lambda = 0$. Otherwise we can replace $\y$ with $\y-a_\lambda\x$, and the proof still holds.

Let $Q(\c) = \|\y-\Z_\A\c\|_2^2/(2n) + \lambda\|\c\|_1$, so that  $\tilde \b_\A = \argmin_{\c \in \mathbb{R}^q} Q(\c)$. Note that $Q(\cdot)$ is strictly convex, and thus $\tilde \b_\A$ is unique, since $\Z_\A$ is full rank by ({\bf A6}). We will show that for all $\xi > 0$ there exists a constant $m$, not depending on $n$, such that
\begin{equation}\label{1:p0}
\lim_{n \rightarrow \infty}\Pr \left[ \inf_{\c : \|\c\|_2 = m}Q(\b_{\lambda\A} +\c\sqrt{q\log(q)/n}) > Q(\b_{\lambda\A})\right] \geq 1-\xi.
\end{equation}
Convexity of $Q$ then implies that $\tilde \b_{\A}$ is in the ball $\{ \b_{\lambda\A} + \c\sqrt{q \log(q)/n} :  \|\c\|_2 \leq m\}$ with probability at least $1-\xi$. Thus, we have $\lim_{n \rightarrow \infty} \Pr[\|\tilde \b_{\A} - \b_{\lambda\A}\|_2 > m \sqrt{q \log(q)/n}]\leq \xi$, i.e.  $\|\tilde \b_{\A} - \b_{\lambda\A}\|_2 = O_p(\sqrt{q \log(q)/n})$.

We now proceed to prove \eqref{1:p0}. Let $\w=\argmin_{\c \in \mathbb{R}^q : \|\c\|_2 = m}Q(\b_{\lambda\A} +\c\sqrt{q\log(q)/n}).$ Expanding terms, we can write
\begin{align}
Q\left(\b_{\lambda\A} +\w\sqrt{q\log(q) /n}\right) -Q (\b_{\lambda\A}) &= -\frac{\sqrt{q \log(q) }}{n^{3/2}}\w^T\Z_\A^T(\y-\Z_\A\b_{\lambda\A}) + \frac{q\log(q)}{2n^2}\w^T\Z_\A^T\Z_\A\w \nonumber \\
&+\lambda\|\b_{\lambda\A} +\w\sqrt{q \log(q) /n}\|_1 - \lambda\|\b_{\lambda\A}\|_1. \label{eq:taylor1}
\end{align}
First note that, for $g,h \in \mathbb{R}$, $|g + h| = |g| + \mathrm{sign}(g)h$ when $|h| \leq |g|$. Thus, we have $\left\|\b_{\lambda\A} + \w\sqrt{q \log(q) /n}\right \|_1 = \|\b_{\lambda\A}\|_1 + \btau_\A^T\w \sqrt{q \log(q)/n}$ when $m\sqrt{q \log(q)/n} < \min\{|b_{\lambda,1}|,\dots,|b_{\lambda,q}|\}=b_{\min}$. Since $\sqrt{q \log(q) /n}/b_{\min} \rightarrow 0$ by ({\bf A6}),  for $n$ large enough we can write 
\begin{align}
Q\left(\b_{\lambda\A} +\w\sqrt{q \log(q) /n}\right) -Q (\b_{\lambda\A}) &=- \frac{\sqrt{q \log(q)}}{n^{3/2}}\w^T\Z_\A^T(\y-\Z_\A\b_{\lambda\A}) + \frac{q \log(q) }{2n^2}\w^T\Z_\A^T\Z_\A\w \nonumber \\
&\quad+\lambda\sqrt{\frac{q \log(q)}{ n}}\btau_\A^T\w \nonumber\\ 
&=- \frac{\sqrt{q \log(q) }}{n^{3/2}}\w^T\left[ \Z_\A^T(\y-\Z_\A\b_{\lambda\A})  - \lambda n\btau_\A\right] \nonumber \\& \quad+ \frac{q \log(q) }{2n^2}\w^T\Z_\A^T\Z_\A\w \label{eq:taylor3} 
\end{align}
We now bound $\w^T\left[ \Z_\A^T(\y-\Z_\A\b_{\lambda\A})  - \lambda n\btau_\A\right]$ in \eqref{eq:taylor3}. Note that $\Z_\A^T(\y-\Z_\A\b_{\lambda\A})- \lambda n\btau_\A =\Z_\A^T\beps$, using \eqref{eq:stationarybl}. Thus we have 
\begin{align}
\left|\w^T\left[ \Z_\A^T(\y-\Z_\A\b_{\lambda\A})  - \lambda n\btau_\A \right] \right|&\leq \|\w\|_1\|\Z_\A^T\beps\|_\infty \nonumber  \\
&\leq \sqrt{q}\|\w\|_2\|\Z_\A^T\beps\|_\infty \label{ml1},
\end{align}
Now note that $\w^T\Z_\A^T\Z_\A\w/n \geq \Lambda_{\min}^2(\bSigma_\A)\|\w\|^2_2$. Thus, we get that
\begin{align*}
Q\left(\b_{\lambda\A} +\w\sqrt{q\log(q)/n}\right) -Q (\b_{\lambda\A}) &\geq -\frac{q \log^{1/2}(q) }{n^{3/2}}\|\w\|_2\|\Z_\A^T\beps\|_\infty  +\frac{q \log(q)}{2n} \Lambda_{\min}^2(\bSigma_\A)\|\w\|^2_2\\ 
&=\frac{q \log^{1/2}(q) m}{n} \left(m\log^{1/2}(q)\Lambda_{\min}^2(\bSigma_\A)/2 - \|\Z_\A^T\beps\|_\infty/\sqrt{n}\right).
\end{align*}
We know $\|\Z_\A^T\beps\|_\infty/\sqrt{n} = O_p(\log^{1/2}(q))$, by Lemma~\ref{basiclemma}, which applies by ({\bf A5}). Thus, we can choose $m$, not depending on $n$, such that \eqref{1:p0} holds, provided that $\Lambda_{\min}(\bSigma_\A)$ is bounded below, which is guaranteed by ({\bf A6}).
\end{proof}


\begin{lemma}\label{1:mainlemma2}
Suppose conditions ({\bf A1}) and ({\bf A3}-{\bf A6}) hold. Then any minimizer $\hat \b_\lambda$ of 
\begin{equation}\label{star}
\|\y -a_\lambda\x- \Z\b\|_2^2/(2n) + \lambda\|\b\|_1
\end{equation}
 satisfies $\|\hat \b_{\lambda\A} - \b_{\lambda\A}\|_2 = O_p\left(\sqrt{q\log(q) / n}\right)$ and $\lim_{n\rightarrow \infty} \Pr[ \|\hat \b_{\lambda\A^c}\|_2 = 0 ] = 1$. 
\end{lemma}
\begin{proof}
As in Lemma~\ref{1:mainlemma}, assume without loss of generality that $a_\lambda = 0$.

It suffices to show that $\tilde \b$, from Lemma~\ref{1:mainlemma}, is the unique minimizer of \eqref{star} with probability tending to 1, since this implies that $\lim_{n\rightarrow \infty} \Pr[ \|\hat \b_{\lambda\A^c}\|_2 = 0 ] = 1$ and that $[n/ (q \log q)]^{1/2}\|\hat \b_{\lambda\A} - \b_{\lambda\A}\|_2 = [n/ (q \log q)]^{1/2}\|\tilde \b_\A - \b_{\lambda\A}\|_2 + o_p(1)$. 

By the Karush-Kuhn-Tucker conditions, $\tilde \b$ is a minimizer of \eqref{star} if and only if $\Z^T(\y-\Z\tilde\b)/n = \lambda\tilde \btau$ for some $\tilde \btau$ satisfying $\|\tilde \btau\|_\infty \leq 1$ and $\tilde \tau_i = \mathrm{sign}(\tilde b_i)$ for $\tilde b_i \neq 0$. Since we already know that $\|\Z_\A^T(\y - \Z\tilde \b)/n\|_\infty \leq \lambda$ and $\z^T_i(\y-\Z\tilde\b)/n = \lambda \mathrm{sign}(\tilde b_i)$ for $\tilde b_i \neq 0$ (by the definition of $\tilde \b$), if we can show that
\begin{align}\label{1:p1}
\lim_{n \rightarrow \infty}\Pr\left[\frac{1}{n}\left \| \Z_{\A^c}^T(\y- \Z\tilde \b)\right\|_\infty < \lambda\right] =1,
\end{align}
then we will have shown that $\tilde \b$ is a minimizer of \eqref{star} with probability tending to 1. Furthermore, when $\| \Z_{\A^c}^T(\y- \Z\tilde \b)/n\|_\infty < \lambda$ holds then $\tilde \b$ is the unique minimizer of $\|\y-\Z\b\|_2^2/(2n) + \lambda\|\b\|_1$. To see this note that \emph{all} minimizers of \eqref{star} produce the same fitted values \citep{tibshirani2013lasso}. Thus, if $\| \Z_{\A^c}^T(\y- \Z\tilde \b)/n\|_\infty < \lambda$, then $\| \Z_{\A^c}^T(\y- \Z\hat\b_\lambda)/n\|_\infty < \lambda$ for any minimizer $\hat \b_\lambda$ of \eqref{star}, which implies that $\|\hat \b_{\lambda\A^c}\|_2 = 0$, by the Karush-Kuhn-Tucker conditions.  When $\|\hat \b_{\lambda\A^c}\|_2 = 0$, then $\hat \b_\lambda=\tilde \b$, which is unique, as was argued in the proof of Lemma~\ref{1:mainlemma}.

We now show that \eqref{1:p1} holds. Adding and subtracting $\Z_{\A^c}^T\Z_\A\b_{\lambda\A}/n$ we get
\begin{align}
\frac{1}{n}\|\Z_{\A^c}^T(\y- \Z\tilde \b)\|_\infty & \leq \frac{1}{n} \|\Z_{\A^c}^T(\y- \Z_\A\b_{\lambda\A})\|_\infty +  \frac{1}{n}\|\Z_{\A^c}^T\Z_\A(\b_{\lambda \A}-\tilde \b_{\A} )\|_\infty \label{ml3}.
\end{align}
First, we bound $(1/n)\|\Z_{\A^c}^T\Z_\A(\b_{\lambda \A}-\tilde \b_{\A} )\|_\infty$ in \eqref{ml3}. By ({\bf A4}) and Lemma~\ref{1:mainlemma}, we have 
\begin{align}
\frac{1}{n}\left\|\Z_{\A^c}^T\Z_\A(\tilde \b_{\A} - \b_{\lambda \A})\right\|_\infty & \leq\frac{1}{n}\left \|\Z_{\A^c}^T\Z_\A\right\|_\infty\| \tilde \b_{\A} - \b_{\lambda \A}\|_\infty  \nonumber\\
& \leq \frac{1}{n}\left\|\Z_{\A^c}^T\Z_\A\right\|_\infty\| \tilde \b_{\A} - \b_{\lambda \A}\|_2 \nonumber \\
& = o_p(1) \label{ml4}.
\end{align}
We now bound  $(1/n) \|\Z_{\A^c}^T(\y- \Z_\A\b_{\lambda\A})\|_\infty$ in \eqref{ml3}. Recall that $\Z^T(\y- \Z\b_\lambda) - n\lambda \btau= \Z^T\beps$. Using Lemma~\ref{basiclemma} to bound $\|\Z_{\A^c}^T\beps\|_\infty/n$ we get that 
\begin{align}
\frac{1}{n}\|\Z_{\A^c}^T(\y- \Z\b_\lambda)\|_\infty & = \frac{1}{n}\|\Z_{\A^c}^T(\y- \Z\b_\lambda) - n\lambda \btau_{\A^c} +n\lambda \btau_{\A^c}\|_\infty \nonumber \\ 
& \leq  \frac{1}{n}\|\Z_{\A^c}^T\beps\|_\infty  + \lambda \|\btau_{\A^c}\|_\infty \nonumber \\
&\leq O_p\left(\frac{\log^{1/2}(d-q)}{n^{1/2}} \right) + \lambda (1-\delta) \nonumber \\
& = o_p(1) + \lambda (1-\delta)  \label{ml5},
\end{align}
where we used $\|\btau_{\A^c}\|_\infty < 1-\delta$, by ({\bf A3}), and $\log^{1/2}(d-q)/n^{1/2} \rightarrow 0$, by ({\bf A6}).

Altogether, applying the bounds \eqref{ml4} and \eqref{ml5} to  \eqref{ml3}, we have 
\begin{align*}
\left \| (1/n)\Z_{\A^c}^T(\y- \Z\tilde \b)\right\|_\infty \leq \lambda (1-\delta)  + o_p(1),
\end{align*}
which is smaller than $\lambda$ with probability tending to 1. Thus, \eqref{1:p1} holds.
\end{proof}


\begin{lemma}\label{1:mainlemma3}
Suppose conditions ({\bf A1}) and ({\bf A3}-{\bf A6}) hold. Then any minimizer $\hat \b_\lambda$ of $\|\y-a_\lambda\x - \Z\b\|_2^2/(2n) + \lambda\|\b\|_1$ satisfies
$$ \sqrt{n}(\hat \b_{\lambda\A} - \b_{\lambda\A}) = \frac{1}{\sqrt{n}}\bSigma_\A^{-1}\Z_\A^T\beps + o_p(1).$$
\end{lemma}
\begin{proof}
As in Lemma~\ref{1:mainlemma}, assume without loss of generality that $a_\lambda = 0$. 

By the stationary conditions defining $\hat \b_\lambda$ we have that
\begin{equation}\label{eq:1:stationaryb}
\Z^T(\y-\Z\hat\b_\lambda)/n = \lambda\hat \btau,
\end{equation}
for some $\hat\btau \in [-1,1]^d$. First, we show that $\Pr[\hat \btau_\A =\btau_\A] \rightarrow 1$ as $n \rightarrow \infty$ (recall from \eqref{eq:stationarybl} that $\btau_\A = \mathrm{sign}(\b_{\lambda\A})$).  By Lemma~\ref{1:mainlemma2}, given $\xi >0$, there exists a $c>0$ such that $\lim _{n \rightarrow \infty}\Pr\left[\|\hat\b_{\lambda\A} -\b_{\lambda\A}\|_2 < c\sqrt{q \log(q)/n}\right] > 1-\xi$. Further, by ({\bf A6}), we have $c\sqrt{q\log(q)/n}<  b_{\min}$ for $n$ large enough. Now, $ b_{\min} > c\sqrt{q\log(q)/n}> \|\hat \b_{\lambda\A} - \b_{\lambda\A}\|_2\geq  \|\hat \b_{\lambda\A} - \b_{\lambda\A}\|_\infty$ implies that all elements of $\hat \b_{\lambda\A}$ are non-zero and that  $\mathrm{sign}(\hat \b_{\lambda\A}) = \mathrm{sign}(\b_{\lambda\A})$.  Thus, since $\hat \btau_\A =\mathrm{sign}(\hat \b_{\lambda\A})$ when all elements of $\hat \b_{\lambda\A} $ are non-zero, we have $\lim_{n\rightarrow \infty}\Pr[ \hat \btau_\A = \btau_\A] > 1-\xi$. Since $\xi$ is arbitrary, we have $\lim_{n\rightarrow \infty}\Pr[\hat \btau_\A =\btau_\A] = 1$.

Thus, from \eqref{eq:1:stationaryb}, we can write
\begin{align}
{\bf 0} &= \Z_\A^T(\y-\Z\hat\b_\lambda) /\sqrt{n} - \sqrt{n}\lambda \hat \btau_\A \nonumber \\
&=  \Z_\A^T(\y-\Z\hat \b_\lambda) /\sqrt{n} - \sqrt{n}\lambda\btau_\A+ o_p(1) \nonumber \\
&= \Z_\A^T(\y-\Z\b_\lambda) /\sqrt{n} - \sqrt{n}\lambda \btau_\A - \frac{1}{n}\Z_\A^T\Z_\A\sqrt{n}(\hat \b_{\lambda\A} - \b_{\lambda\A}) \nonumber \\
&\qquad- \frac{1}{n}\Z_\A^T\Z_{\A^c}\sqrt{n}\hat \b_{\lambda\A^c} + o_p(1) \label{lemma3.1}. 
\end{align}
Now, from Lemma~\ref{1:mainlemma2} we have that $\Pr[\|\hat \b_{\lambda\A^c}\|_2 =0] \rightarrow 1$.  Also, from \eqref{eq:stationarybl} we know that $\Z_\A^T(\y-\Z\b_\lambda) /\sqrt{n} - \sqrt{n}\lambda \btau_\A = \Z_\A^T\beps/\sqrt{n}$.
Thus,  we can write
\begin{equation*}
{\bf 0}=\Z_\A^T\beps/\sqrt{n} - \frac{1}{n}\Z_\A^T\Z_\A\sqrt{n}(\hat \b_{\lambda\A} - \b_{\lambda\A}) + o_p(1).
\end{equation*}
Multiplying through by $\bSigma_\A^{-1}$ gives the result.
\end{proof}


With these lemmas, we can now prove our main result.
\begin{proof}[Proof of Proposition~\ref{thm:lassodist}]
Recall that $\hat \b^0_\lambda = \argmin_\b\{\|\y-\Z\b\|_2^2/(2n) + \lambda \|\b\|_1 \}$. Since $a_\lambda=0$, we have that $\hat \b_\lambda^0 =\hat \b_\lambda $, using the notation of Lemmas~\ref{1:mainlemma2} and \ref{1:mainlemma3}.

First, we have
\begin{align*}
T_\lambda &= \frac{1}{\sqrt{n}} \x^T(\y-\Z\hat\b_\lambda)  \\ 
&=\frac{1}{\sqrt{n}} \x^T(\y-\Z\b_\lambda) -\frac{1}{n}\x^T\Z_\A\sqrt{n}(\hat \b_{\lambda\A} - \b_{\lambda\A}) - \frac{1}{\sqrt{n}} \x^T\Z_{\A^c}\hat \b_{\lambda\A^c}.
\end{align*}
Now, $\Pr[\|\hat \b_{\lambda\A^c}\|_2=0] \rightarrow 1$ by Lemma~\ref{1:mainlemma2}. Thus, we get
$$
T_\lambda =\frac{1}{\sqrt{n}} \x^T(\y-\Z\b_\lambda) -  \bsigma^T_{x\A}\sqrt{n}(\hat \b_{\lambda\A} - \b_{\lambda\A}) +o_p(1),
$$
where $\bsigma_{x\A} = \Z_\A^T\x/n$.

Now, using Lemma~\ref{1:mainlemma3}, we know  
$$
\sqrt{n}\bsigma_{x\A}^T(\hat \b_{\lambda\A} - \b_{\lambda\A})=\frac{1}{\sqrt{n}}\bsigma_{x\A}^T\bSigma_\A^{-1} \Z_\A^T\beps  + o_p(1).
$$ Also, since $H_{0\lambda}: a_\lambda =0$ is true, we have $\x^T(\y-\Z\b_\lambda) = \x^T\beps$, and so 
$$
T_\lambda= \frac{1}{\sqrt{n}} \left(\x^T\beps-\bsigma_{x\A}^T\bSigma_\A^{-1}\Z_\A^T\beps\right)  +o_p(1)= \frac{1}{\sqrt{n}}\x^T({\bf I}_n - {\bf P}_\A)\beps  + o_p(1).
$$ Dividing by $\sigma_\beps\sqrt{\x^T({\bf I}_n - {\bf P}_\A)\x/n}$ we get
\begin{align*}
\frac{T_\lambda}{\sigma_\beps\sqrt{\x^T({\bf I}_n - {\bf P}_\A)\x/n}} &= \frac{{\bf r}^T\beps}{\sigma_\beps\|{\bf r}\|_2} + o_p(1),
\end{align*}
where  ${\bf r}^T = \x^T({\bf I}_n - {\bf P}_\A)$. Now, the Lindeberg-Feller Central Limit Theorem guarantees that $T_\lambda/\left(\sigma_\beps\sqrt{\x^T({\bf I}_n - {\bf P}_\A)\x/n}\right) \rightarrow_d N(0,1)$ if the Lindeberg condition holds: 
$$
\lim_{n \rightarrow \infty} \sum_{i=1}^n\E \left[\, \frac{ (r_i \epsilon_i)^2}{\sigma_\beps^2\|{\bf r}\|^2_2} 1\left\{ \frac{|r_i \epsilon_i|}{\sigma_\beps\|{\bf r}\|_2} >  \eta \right \} \,\right]= 0, \quad \forall \eta > 0.
$$
Using that $|r_i| \leq \|{\bf r}\|_\infty$, and that the $\epsilon_i$'s are identically distributed, we get
\begin{align*} 
\sum_{i=1}^n&\E \left[\, \frac{ (r_i \epsilon_i)^2}{\sigma_\beps^2\|{\bf r}\|^2_2} 1\left\{ \frac{|r_i \epsilon_i|}{\sigma_\beps\|{\bf r}\|_2} >  \eta \right \} \,\right] \\
&\leq \sum_{i=1}^n \frac{r_i^2 }{\sigma_\beps^2\|{\bf r} \|_2^2} \E\left[ \epsilon_i^21\left\{ \frac{| \epsilon_i|\|{\bf r}\|_\infty }{\sigma_\beps\|{\bf r}\|_2 } >  \eta \right \} \right] \\
& = \frac{1}{\sigma_\beps^2} \E\left[ \epsilon_1^2 1\left\{ \frac{| \epsilon_1|\|{\bf r}\|_\infty }{\sigma_\beps\|{\bf r}\|_2 } >  \eta \right \} \right].
\end{align*}
Now, since $\|{\bf r}\|_\infty/\|{\bf r}\|_2\rightarrow 0$ by ({\bf A2}), we have that $\epsilon_1^2 1\left\{ | \epsilon_1|\|{\bf r}\|_\infty/( \sigma_\beps\|{\bf r}\|_2)  >  \eta \right \} \rightarrow_p 0$. Thus we can apply the Dominated Convergence Theorem, using the dominating random variable $\epsilon_1^2$, which satisfies $\E[\epsilon_1^2] =\sigma^2_\epsilon < \infty$ and  $\epsilon_1^2 \geq \epsilon_1^2 1\left\{ | \epsilon_1|\|{\bf r}\|_\infty/( \sigma_\beps\|{\bf r}\|_2)  >  \eta \right \}$ with probability 1, to get that 
$$\frac{1}{\sigma_\beps^2} \E\left[ \epsilon_1^2 1\left\{ \frac{| \epsilon_1|\|{\bf r}\|_\infty }{\sigma_\beps\|{\bf r}\|_2 } >  \eta \right \} \right]\rightarrow 0,$$
which in turn gives the result.
\end{proof}

\bibliographystyle{apalike}

\end{document}